\documentclass[12pt,draftcls,onecolumn]{IEEEtran}
\usepackage{amsfonts}
\usepackage{cite,graphicx,amsmath,amssymb}
\usepackage{subfigure}
\usepackage{citesort}
\usepackage{fancyhdr}
\usepackage{dsfont}
 
\usepackage{array,color}
\usepackage{amsmath}
\newtheorem{remark}{Remark}

\newtheorem{challenge}{Challenge}

\newcommand{\dis}{\stackrel{d}{\sim}}
\newcommand{\eqla}{\stackrel{(a)}{=}}
\newcommand{\eqlb}{\stackrel{(b)}{=}}

\newcommand{\appa}{\stackrel{(a)}{\approx}}

\newtheorem{theorem}{Theorem}

\newtheorem{lemma}{Lemma}

\newtheorem{corollary}{Corollary}

\newtheorem{proposition}{Proposition}

\newcommand{\define}{\stackrel{\Delta}{=}}

\newcommand{\papertitle}{Analysis and Optimization of Inter-tier Interference Coordination in Downlink Multi-Antenna HetNets with Offloading}

\IEEEoverridecommandlockouts
\begin{document}

\title{\papertitle}
\author{Yueping Wu$^{*}$,\hspace{5mm}Ying Cui$^{\ddag}$,\hspace{5mm}Bruno Clerckx$^{*\dagger}$\\
\vspace{2mm}
{\footnotesize $^{*}$Dept. of Electrical and Electronic Engineering, Imperial College London, $^{\dagger}$School of Electrical Engineering, Korea University\\
$^{\ddag}$Dept. of Electronic Engineering, Shanghai Jiao Tong University\\
}
\thanks{The work of Y.~Wu and B.~Clerckx was partially supported by the Seventh Framework Programme for Research of the European Commission under grant number HARP-318489. This work has been submitted in part to IEEE ICC 2015.}}
\maketitle \setcounter{page}{1} 

\begin{abstract}
Heterogeneous networks (HetNets) with offloading is considered as an effective way to meet the high data rate demand of future wireless service. However, the offloaded users suffer from strong inter-tier interference, which reduces the benefits of offloading and is one of the main limiting factors of the system performance. In this paper, we investigate an interference nulling (IN) scheme in improving the system performance by carefully managing the inter-tier interference to the offloaded users in downlink two-tier HetNets with multi-antenna base stations. Utilizing tools from stochastic geometry, we first derive a tractable expression for the rate coverage probability of the IN scheme. Then, by studying its order, we obtain the optimal design parameter, i.e., the degrees of freedom that can be used for IN, to maximize the rate coverage probability. Finally, we analyze the rate coverage probabilities of the simple offloading scheme without interference management and the multi-antenna version of the almost blank subframes (ABS) scheme in 3GPP LTE, and compare the performance of the IN scheme with these two schemes. Both analytical and numerical results show that the IN scheme can achieve good performance gains over both of these two schemes, especially in the large antenna regime. 
\end{abstract}
\vspace{1mm}
\begin{keywords}
Heterogeneous networks, offloading, multiple antennas, inter-tier interference coordination, rate coverage probability, stochastic geometry, optimization.
\end{keywords}


\section{Introduction}\label{sec:intro}
The modern wireless networks have seen a significant increase in the number of users and the scope of high data rate applications. The growth of data rate demand is expected to continue for at least a few more years \cite{cisco14}. The conventional cellular solution, which comprises of high power base stations (BSs), each covering a large cellular area, will not be able to scale with the increasing data rate demand. A promising solution is the deployment of low power small cell nodes overlaid with high power macro-BSs, so called heterogeneous networks (HetNets). HetNets are capable of aggressively reusing existing spectrum assets to support high data rate applications. Due to the large power at macro-BSs, most of the users intend to connect with macro-BSs, which causes the problem of \emph{load imbalancing} \cite{jo12}. To address load imbalancing, some users are offloaded to the lightly loaded small cells via a bias factor \cite{andrews14}. The performance of HetNets with offloading has been investigated in various literature (see e.g., \cite{jo12,singh13may}). However, in HetNets with offloading, the offloaded users (i.e., the users offloaded from the macro-cell tier to the small-cell tier via bias) have degraded signal-to-interference ratio (SIR), which is one of the limiting factors of the network performance. Interference management techniques are thus desired in HetNets with offloading. One such technique is almost blank subframes (ABS) in 3GPP LTE  \cite{damnjanovic11}. In ABS, (time or frequency) resource is partitioned, whereby the offloaded users and the other users are served using different portions of the resource. The performance of ABS in HetNets with offloading was analyzed in \cite{singh13} using tools from stochastic geometry. Another interference management technique was proposed for single-antenna HetNets in \cite{sakr14} to reduce the interference to each offloaded user by cooperation between its nearest macro-BS and nearest pico-BS. Under the scheme in \cite{sakr14}, the scheduled offloaded user and the users of its nearest macro-BS cannot be served using the same resource. Note that \cite{singh13,sakr14} considered single-antenna HetNets, and both schemes studied in \cite{singh13,sakr14} may not fully utilize the system resource. 

Deploying multiple antennas in HetNets can further improve data rates for future wireless service. With multiple antennas, more effective interference management techniques can be implemented. For example, references \cite{Hosseini13,Kountouris13,AdhikaryITA14,Adhikary14} investigated the performance of a HetNet with a single multi-antenna macro-BS and multiple small-BSs, where the multiple antennas at the macro-BS are used for serving its scheduled users as well as mitigating interference to the receivers in small cells using different interference coordination schemes. These schemes have been analyzed and shown to have performance improvement. However, since only one macro-BS is considered, the analytical results obtained in \cite{Hosseini13,Kountouris13,AdhikaryITA14,Adhikary14} cannot reflect the macro-tier interference, and thus cannot offer accurate insights for practical HetNets. In \cite{xia13}, interference coordination among a \emph{fixed} number of neighboring BSs was investigated in downlink large multi-antenna HetNets. However, this scheme may not fully exploit the spatial properties of the interference in large HetNets, and thus cannot effectively improve the system performance. Moreover, offloading was not considered in \cite{xia13}. So far, it is still not clear how the interference coordination schemes and the system parameters affect the performance of large multi-antenna HetNets with offloading.


In this paper, we consider offloading in downlink two-tier large stochastic multi-antenna HetNets where a macro-cell tier is overlaid with a pico-cell tier, and investigate an interference nulling (IN) scheme in improving the performance of the offloaded users. The IN scheme has a design parameter, which is the degree of freedom $U$ that can be used at each macro-BS for avoiding its interference to some of its offloaded users. In particular, each macro-BS utilizes the low-complexity zero-forcing beamforming (ZFBF) precoder to suppress interference to at most $U$ offloaded users as well as boost the signal to its scheduled user. Interference coordination using beamforming technique in large stochastic HetNets causes spatial dependence among macro-BSs and pico-BSs \cite{Adhikary14}, and user dependence among offloaded users. Thus, it is more challenging to analyze than interference coordination in multi-antenna stochastic \emph{single-tier} cellular networks \cite{zhang14,lee14IC,li14IC}. In this paper, by adopting appropriate approximations and utilizing tools from stochastic geometry, we first present a tractable expression for the rate coverage probability of the IN scheme. To our best knowledge, this is the first work analyzing the interference coordination technique in large stochastic multi-antenna HetNets with offloading. To further improve the rate coverage probability of the IN scheme, we consider the optimization of its design parameter. Note that optimization problems in large HetNets with single-antenna BSs were investigated in \cite{nie14,lin14}. The objective functions in \cite{nie14,lin14} are relatively simple, and bounds of the objective function and the constraint are utilized to obtain near-optimal solutions. The optimization problem in large multi-antenna HetNets we consider is an integer programming problem with a very complicated objective function. Hence, it is quite challenging to obtain the optimal solution. First, for the asymptotic scenario where the rate threshold is small, by studying the order behavior of the rate coverage probability, we prove that the optimal design parameter converges to a fixed value, which equals to either the antenna number difference between each maco-BS and each pico-BS or the antenna number difference minus one. Next, for the general scenario, we show that besides the number of antennas, the optimal design parameter also depends on other system parameters. 

Finally, we compare the IN scheme with the simple offloading scheme without interference management and the multi-antenna version of ABS in 3GPP LTE. In particular, we first analyze the rate coverage probabilities of the simple offloading scheme and ABS. Then, we compare the IN scheme with the simple offloading scheme and ABS, respectively, in terms of the rate coverage probability of each user type and the overall rate coverage probability. Both the analytical and numerical results show that the IN scheme can achieve good rate coverage probability gains over both of these two schemes, especially in the large antenna regime.

\section{System Model}\label{sec:sys_model}
\subsection{Downlink Two-Tier Heterogeneous Networks}
We consider a downlink two-tier HetNet where a macro-cell tier is overlaid with a pico-cell tier, as shown in Fig.\ \ref{fig:wvoronoi}. The locations of the macro-BSs and the pico-BSs are spatially distributed as two independent Homogeneous Poisson point processes (PPPs) $\Phi_{1}$ and $\Phi_{2}$ with densities $\lambda_{1}$ and $\lambda_{2}$, respectively. The locations of the users are also distributed as an independent homogeneous PPP $\Phi_{u}$ with density $\lambda_{u}$. Without loss of generality (w.l.o.g.), denote the macro-cell tier as the $1$st tier and the pico-cell tier as the $2$nd tier. We focus on the downlink scenario. Each macro-BS has $N_{1}$ antennas with total transmission power $P_{1}$, each pico-BS has $N_{2}$ antennas with total transmission power $P_{2}$, and each user has a single antenna. We consider both large-scale fading and small-scale fading. Specifically, due to large-scale fading, transmitted signals (from the $j$th tier) with distance $r$ are attenuated by a factor $\frac{1}{r^{\alpha_{j}}}$ ($j=1,2$), where $\alpha_{j}>2$ is the path loss exponent of the $j$th tier. For small-scale fading, we assume Rayleigh fading channels. 

\subsection{User Association}
We assume open access \cite{jo12}. As discussed in Section \ref{sec:intro}, due to the larger power at the macro-BSs, the load imbalancing problem arises if the user association is only according to the long-term average received power (RP). To remit the load imbalancing problem, the bias factor $B_{j}$ ($j=1,2$) is introduced to tier $j$, where $B_{2}>B_{1}$, to offload users from the heavily loaded macro-cell tier to the lightly loaded pico-cell tier. Specifically, user $i$ (denoted as $u_{i}$) is associated with the BS which provides the maximum \emph{long-term average} biased-received-power (BRP) (among all the macro-BSs and pico-BSs). Here, the long-term average BRP is defined as the average RP multiplied by a bias factor. This associated BS is called the \emph{serving BS} of user $i$. Note that within each tier, the nearest BS to user $i$ provides the strongest long-term average BRP in this tier. User $i$ is thus associated with the nearest BS in the $j^{*}_{i}$th tier if\footnote{In the user association procedure, the first antenna is normally used to transmit signal (using the total transmission power of each BS) for BRP determination according to LTE standards \cite{sesia09}.} 
\begin{align}\label{eq:tier_select}
j_{i}^{*}&=  {\arg\:\max}_{j\in\{1,2\}}P_{j}B_{j}Z_{i,j}^{-\alpha_{j}}
\end{align}
where $Z_{i,j}$ is the distance between user $i$ and its nearest BS in the $j$th tier. We observe that, for given $\{P_{j}\}$, $\{Z_{i,j}\}$ and $\{\alpha_{j}\}$, user association is only affected by the ratio between $B_{1}$ and $B_{2}$. Thus, w.l.o.g., we assume $B_{1}=1$ and $B_{2}=B>1$. After user association, each BS schedules its associated users according to TDMA, i.e., scheduling one user in each time slot, so that there is no intra-cell interference. 

According to the above mentioned user association policy and the offloading strategy, all the users can be partitioned into the following three disjoint user sets: 
\begin{enumerate}
\item the set of \emph{macro-users}: $\mathcal{U}_{1}=\left\{u_{i}|P_{1}Z_{i,1}^{-\alpha_{1}}\ge BP_{2}Z_{i,2}^{-\alpha_{2}}\right\}$\;,
\item the set of \emph{unoffloaded pico-users}: $\mathcal{U}_{2\bar{O}}=\left\{u_{i}|P_{2}Z_{i,2}^{-\alpha_{2}}>P_{1}Z_{i,1}^{-\alpha_{1}}\right\}$\;,
\item the set of \emph{offloaded users}: $\mathcal{U}_{2O}=\left\{u_{i}|P_{2}Z_{i,2}^{-\alpha_{2}}\le P_{1}Z_{i,1}^{-\alpha_{1}}<BP_{2} Z_{i,2}^{-\alpha_{2}}\right\}$\;,
\end{enumerate}
where the macro-users are associated with the maco-cell tier, the unoffloaded pico-users are associated with the pico-cell tier (even without bias), and the offloaded users are offloaded from the macro-cell tier to the pico-cell tier (due to bias $B>1$), as illustrated in Fig.\ \ref{fig:user_set}. Moreover, $\mathcal{U}_{2}=\mathcal{U}_{2\bar{O}}\bigcup\mathcal{U}_{2O}$ represents \emph{the set of pico-users}. 


\begin{figure}[t]
\centering
\subfigure[System Model ($U=1$)]{
\includegraphics[width=3.6in]
{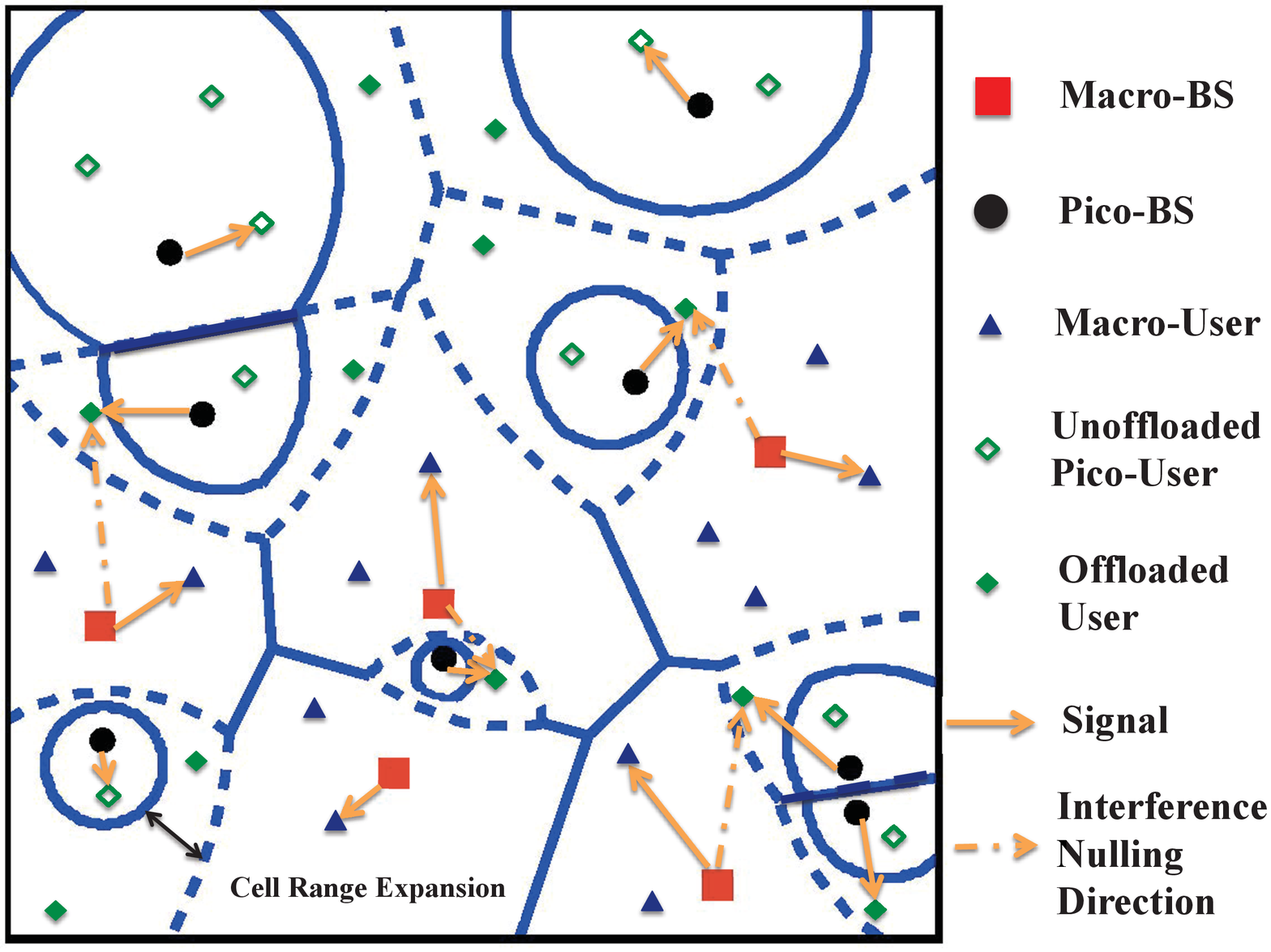}
\label{fig:wvoronoi}
}
\subfigure[User Set Illustration]{
\includegraphics[width=2.6in]{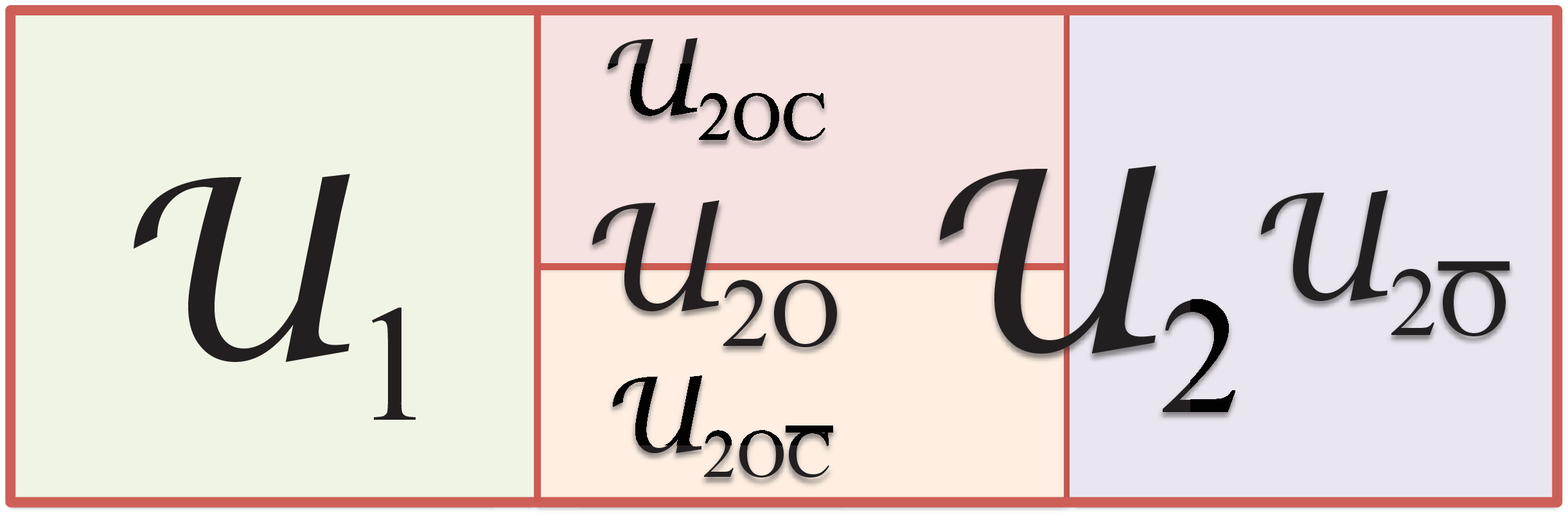}
\label{fig:user_set}
}
\caption{\small System model and user set illustration.}\label{fig:model}
\vspace{-8mm}
\end{figure}

\subsection{Performance Metric}
In this paper, we study the performance of the typical user denoted as\footnote{The index of the typical user and its serving BS is $0$.} $u_{0}$, which is located at the origin and is scheduled \cite{haenggi09}. Since HetNets are interference-limited, in this paper, we ignore the thermal noise in the analysis, as in \cite{heath13}. Note that the analytical results with thermal noise can be calculated in a similar way. We investigate the \emph{rate coverage probability} of the typical user, which is defined as the probability that the rate of the typical user is larger than a threshold \cite{singh13,singh13may}. Specifically, let $R_{0}=\frac{W}{L_{0}}\log_{2}\left(1+{\rm SIR}_{0}\right)$ denote the rate of the typical user, where $W$ is the available resource (e.g., time or frequency), $L_{0}$ is the total number of associated users (i.e., \emph{load}) of the typical user's serving BS, and ${\rm SIR}_{0}$ is the SIR of the typical user. Then, the rate coverage probability can be mathematically written as
\begin{align}\label{eq:CPrate_def}
\mathcal{R}(\tau)&\define{\rm Pr}\left(R_{0}>\tau\right)={\rm Pr}\left(\frac{W}{L_{0}}\log_{2}\left(1+{\rm SIR}_{0}\right)>\tau\right)
\end{align} where $\tau$ is the rate threshold. Note that $R_{0}$ is a random variable with randomness induced by ${\rm SIR}_{0}$ and $L_{0}$. Thus, the rate coverage probability captures the effects of the distributions of both ${\rm SIR}_{0}$ and $L_{0}$ \cite{singh13}. The rate coverage probability is suitable for applications with strict rate requirement, e.g., video services \cite{singh13may}.


\section{Inter-tier Interference Nulling}\label{sec:ITIC_model}
In HetNets with offloading, the offloaded users normally suffer from stronger interference than the macro-users and unoffloaded pico-users. \footnote{For each offloaded user, its nearest macro-BS, which provides the strongest long-term average RP, now becomes the dominant interferer of this offloaded user. However, for each macro-user or unoffloaded pico-user, the BS which provides the strongest long-term average RP is its serving BS. Therefore, the offloaded users suffer the strongest interference.} The dominant interference to each offloaded user, caused by its nearest macro-BS \cite{singh13may}, is one of the limiting factors of the system performance. In this section, we first elaborate on an inter-tier IN scheme to avoid the dominant interference to the offloaded users, so as to improve the system performance. Then, we obtain some results on the distributions of some related random variables of this scheme.

\subsection{IN Scheme Description}\label{sub:IN_intro}
We now describe an inter-tier IN scheme to avoid the dominant interference to the offloaded users by making use of at most $U$ ($U<N_{1}$) DoF at each macro-BS which has $N_{1}$ antennas. In particular, we use the low-complexity ZFBF precoder at each macro-BS to perform inter-tier IN. Note that $U$ is the design parameter of this scheme. When $U=0$, the IN scheme reduces to the simple offloading scheme without interference management. We first introduce several types of users related to this scheme. For each macro-BS, we refer to the users offloaded from it to their nearby pico-BSs as the \emph{offloaded users} of this macro-BS. All these offloaded users may not be scheduled by their nearest pico-BSs simultaneously, as each BS schedules one user in each time slot. In each time slot, we refer to the offloaded users scheduled by their nearest pico-BSs as \emph{active offloaded users} (of this slot). In the IN scheme, each macro-BS avoids its interference to some of its active offloaded users in a particular time slot, which are referred to as the \emph{IN offloaded users} of this macro-BS. We refer to the remaining offloaded users as \emph{non-IN offloaded users}. Hence, under the IN scheme, in a particular time slot, the offloaded users $\mathcal{U}_{2O}$ are further divided into two sets, i.e., $\mathcal{U}_{2O}=\mathcal{U}_{2OC}\bigcup \mathcal{U}_{2O\bar{C}}$, where $\mathcal{U}_{2OC}$ denotes the IN offloaded user set and $\mathcal{U}_{2O\bar{C}}$ denotes the non-IN offloaded user set. Note that under the IN scheme, the users can be partitioned into four disjoint user sets, namely, $\mathcal{U}_{1}$, $\mathcal{U}_{2\bar{O}}$, $\mathcal{U}_{2OC}$ and $\mathcal{U}_{2O\bar{C}}$, as illustrated in Fig.\ \ref{fig:user_set}. 

Next, we discuss how to determine the IN offloaded users of each macro-BS. Specifically, let $U_{2O_{a},\ell}$ denote the number of active offloaded users of macro-BS $\ell$, each of which is scheduled by a different pico-BS.  If $U_{2O_{a},\ell}\le U$, macro-BS $\ell$ can perform IN to all of its $U_{2O_{a},\ell}$ active offloaded users using $U_{2O_{a},\ell}$ DoF. However, if $U_{2O_{a},\ell}>U$, macro-BS $\ell$ randomly selects $U$ out of $U_{2O_{a},\ell}$ active offloaded users according to the uniform distribution to perform IN using $U$ DoF. Hence, macro-BS $\ell$ performs IN to $u_{2OC,\ell}\define\min\left(U,U_{2O_{a},\ell}\right)$ out of $U_{2O_{a},\ell}$ active offloaded users. Note that the DoF used for IN (referred to as IN DoF) at macro-BS $\ell$ is $u_{2OC,\ell}$. All the remaining $N_{1}-u_{2OC,\ell}$ DoF at macro-BS $\ell$ are used for boosting the signal to its scheduled user. 

Now, we introduce the precoding vectors at macro-BSs and pico-BSs in the IN scheme, respectively. First, each macro-BS utilizes the low-complexity ZFBF precoder to serve its scheduled user and simultaneously perform IN to its IN offloaded users. Specifically,  denote $\mathbf{H}_{1,\ell}=\left[\mathbf{h}_{1,\ell}\; \mathbf{g}_{1,\ell1}\;\ldots\;\mathbf{g}_{1,\ell u_{2OC,\ell}}\right]^{\dagger}$, where\footnote{The notation $X \dis Y$ means that $X$ \emph{is distributed as} $Y$.} $\mathbf{h}_{1,\ell}\dis\mathcal{CN}_{N_{1},1}\left(\mathbf{0}_{N_{1}\times 1},\mathbf{I}_{N_{1}}\right)$ denotes the channel vector between macro-BS $\ell$ and its scheduled user, and $\mathbf{g}_{1,\ell i}\dis \mathcal{CN}_{N_{1},1}\left(\mathbf{0}_{N_{1}\times 1},\mathbf{I}_{N_{1}}\right)$ denotes the channel vector between macro-BS $\ell$ and its IN offloaded user $i$ $(i=1,\ldots,u_{2OC,\ell})$. The ZFBF precoding matrix at macro-BS $\ell$ is designed to be $\mathbf{W}_{1,\ell}=\mathbf{H}_{1,\ell}^{\dagger}\left(\mathbf{H}_{1,\ell}\mathbf{H}_{1,\ell}^{\dagger}\right)^{-1}$ and the ZFBF vector at macro-BS $\ell$ is designed to be $\mathbf{f}_{1,\ell}=\frac{\mathbf{w}_{1,\ell}}{\|\mathbf{w}_{1,\ell}\|}$, where $\mathbf{w}_{1,\ell}$ is the first column of $\mathbf{W}_{1,\ell}$. Next, each pico-BS utilizes the maximal ratio transmission precoder to serve its scheduled user. Specifically, the beamforming vector at pico-BS $\ell$ is $\mathbf{f}_{2,\ell}=\frac{\mathbf{h}_{2,\ell}}{\left\|\mathbf{h}_{2,\ell}\right\|}$, 
where $\mathbf{h}_{2,\ell}\dis\mathcal{CN}_{N_{2},1}\left(\mathbf{0}_{N_{2}\times 1},\mathbf{I}_{N_{2}}\right)$ denotes the channel vector between pico-BS $\ell$ and its scheduled user.

We now discuss the received signal and the corresponding SIR of the typical user $u_{0}\in\mathcal{U}_{k}$ ($k\in\mathcal{K}\define\{1,2\bar{O},2OC,2O\bar{C}\}$). 
\subsubsection{Macro-User}\label{subsec:macro}
The received signal and SIR of the typical user $u_{0}\in \mathcal{U}_{1}$ are\footnote{In this paper, all macro-BSs and pico-BSs are assumed to be active. The same assumption can also be seen in the existing papers (see e.g., \cite{singh13,andrews11}).}
\begin{align}
&y_{1,0}=\frac{1}{Y_{1}^{\frac{\alpha_{1}}{2}}}\mathbf{h}_{1,00}^{\dagger}\mathbf{f}_{1,0}x_{1,0}+\sum_{\ell\in\Phi\left(\lambda_{1}\right)\backslash B_{1,0}}\frac{1}{\left|D_{1,\ell0}\right|^{\frac{\alpha_{1}}{2}}}\mathbf{h}_{1,\ell0}^{\dagger}\mathbf{f}_{1,\ell}x_{1,\ell}+\sum_{\ell\in\Phi\left(\lambda_{2}\right)}\frac{1}{\left|D_{2,\ell0}\right|^{\frac{\alpha_{2}}{2}}}\mathbf{h}^{\dagger}_{2,\ell0}\mathbf{f}_{2,\ell}x_{2,\ell}\;,\label{eq:y0_macro}\\
&{\rm SIR}_{{\rm IN},1,0} = \frac{\frac{P_{1}}{Y_{1}^{\alpha_{1}}}\left|\mathbf{h}_{1,00}^{\dagger}\mathbf{f}_{1,0}\right|^{2}}{P_{1}\sum_{\ell\in\Phi\left(\lambda_{1}\right)\backslash B_{1,0}}\frac{1}{\left|D_{1,\ell0}\right|^{\alpha_{1}}}\left|\mathbf{h}_{1,\ell0}^{\dagger}\mathbf{f}_{1,\ell}\right|^{2}+P_{2}\sum_{\ell\in\Phi\left(\lambda_{2}\right)}\frac{1}{\left|D_{2,\ell0}\right|^{\alpha_{2}}}\left|\mathbf{h}_{2,\ell0}^{\dagger}\mathbf{f}_{2,\ell}\right|^{2}}\label{eq:SINR_marco}
\end{align} 
where $B_{1,0}$ is the serving macro-BS of $u_{0}$, $Y_{1}$ is the distance between $u_{0}$ and $B_{1,0}$, $\left|D_{j,\ell0}\right|$ $(j=1,2)$ is the distance from BS $\ell$ in the $j$th tier to $u_{0}$, $x_{1,\ell}$ is the symbol sent from macro-BS $\ell$ to its scheduled user satisfying ${\rm E}\left[x_{1,\ell}x_{1,\ell}^{*}\right]=P_{1}$, and $x_{2,\ell}$ is the symbol sent from pico-BS $\ell$ to its scheduled user satisfying ${\rm E}\left[x_{2,\ell}x_{2,\ell}^{*}\right]=P_{2}$. Here, $\left|\mathbf{h}_{1,00}^{\dagger}\mathbf{f}_{1,0}\right|^{2}\dis{\rm Gamma}\left(N_{1}-u_{2OC,0},1\right)$, $\left|\mathbf{h}_{1,\ell0}^{\dagger}\mathbf{f}_{1,\ell}\right|^{2}\dis{\rm Gamma}(1,1)$, and $\left|\mathbf{h}_{2,\ell0}^{\dagger}\mathbf{f}_{2,\ell}\right|^{2}\dis{\rm Gamma}(1,1)$. 

\subsubsection{Unoffloaded Pico-User}\label{subset:pico}
The received signal and SIR of the typical user $u_{0}\in\mathcal{U}_{2\bar{O}}$ are
\begin{align}
&y_{2\bar{O},0}=\frac{1}{Y_{2}^{\frac{\alpha_{2}}{2}}}\mathbf{h}_{2,00}^{\dagger}\mathbf{f}_{2,0}x_{2,0}+\sum_{\ell\in\Phi\left(\lambda_{1}\right)}\frac{1}{\left|D_{1,\ell0}\right|^{\frac{\alpha_{1}}{2}}}\mathbf{h}_{1,\ell0}^{\dagger}\mathbf{f}_{1,\ell}x_{1,\ell}+\sum_{\ell\in\Phi\left(\lambda_{2}\right)\backslash B_{2,0}}\frac{1}{\left|D_{2,\ell0}\right|^{\frac{\alpha_{2}}{2}}}\mathbf{h}_{2,\ell0}^{\dagger}\mathbf{f}_{2,\ell}x_{2,\ell}\;,\label{eq:y0_pico}\\
&{\rm SIR}_{{\rm IN},2\bar{O},0} = \frac{\frac{P_{2}}{Y_{2}^{\alpha_{2}}}\left|\mathbf{h}_{2,00}^{\dagger}\mathbf{f}_{2,0}\right|^{2}}{P_{1}\sum_{\ell\in\Phi\left(\lambda_{1}\right)}\frac{1}{\left|D_{1,\ell0}\right|^{\alpha_{1}}}\left|\mathbf{h}_{1,\ell0}^{\dagger}\mathbf{f}_{1,\ell}\right|^{2}+P_{2}\sum_{\ell\in\Phi\left(\lambda_{2}\right)\backslash B_{2,0}}\frac{1}{\left|D_{2,\ell0}\right|^{\alpha_{2}}}\left|\mathbf{h}_{2,\ell0}^{\dagger}\mathbf{f}_{2,\ell}\right|^{2}}\label{eq:SINR_pico}
\end{align} 
where $B_{2,0}$ is the serving pico-BS of $u_{0}$, and $Y_{2}$ is the distance between $u_{0}$ and $B_{2,0}$. Here, $\left|\mathbf{h}_{2,00}^{\dagger}\mathbf{f}_{2,0}\right|^{2}\dis{\rm Gamma}\left(N_{2},1\right)$. 

\subsubsection{IN Offloaded User}\label{subset:BC}
When $u_{0}\in\mathcal{U}_{2OC}$, the typical user $u_{0}$ does not suffer interference from its nearest macro-BS. Thus, the received signal and SIR of $u_{0}\in\mathcal{U}_{2OC}$ are
\begin{align}
&y_{2OC,0}=\frac{1}{Y_{2}^{\frac{\alpha_{2}}{2}}}\mathbf{h}_{2,00}^{\dagger}\mathbf{f}_{2,0}x_{2,0}+\sum_{\ell\in\Phi\left(\lambda_{1}\right)\backslash B_{1,0}}\frac{\mathbf{h}_{1,\ell0}^{\dagger}\mathbf{f}_{1,\ell}}{\left|D_{1,\ell0}\right|^{\frac{\alpha_{1}}{2}}}x_{1,\ell}+\sum_{\ell\in\Phi\left(\lambda_{2}\right)\backslash B_{2,0}}\frac{\mathbf{h}_{2,\ell0}^{\dagger}\mathbf{f}_{2,\ell}}{\left|D_{2,\ell0}\right|^{\frac{\alpha_{2}}{2}}}x_{2,\ell}\;,\label{eq:y0_BC}\\
&{\rm SIR}_{{\rm IN},2OC,0} = \frac{\frac{P_{2}}{Y_{2}^{\alpha_{2}}}\left|\mathbf{h}_{2,00}^{\dagger}\mathbf{f}_{2,0}\right|^{2}}{P_{1}\sum_{\ell\in\Phi\left(\lambda_{1}\right)\backslash B_{1,0}}\frac{1}{\left|D_{1,\ell0}\right|^{\alpha_{1}}}\left|\mathbf{h}_{1,\ell0}^{\dagger}\mathbf{f}_{1,\ell}\right|^{2}+P_{2}\sum_{\ell\in\Phi\left(\lambda_{2}\right)\backslash B_{2,0}}\frac{1}{\left|D_{2,\ell0}\right|^{\alpha_{2}}}\left|\mathbf{h}_{2,\ell0}^{\dagger}\mathbf{f}_{2,\ell}\right|^{2}}\;.\label{eq:SINR_BC}
\end{align} 

\subsubsection{Non-IN Offloaded User}\label{subsec:BbarC}
When $u_{0}\in\mathcal{U}_{2O\bar{C}}$, the typical user $u_{0}$ is not selected for IN, and thus it still suffers interference from its nearest macro-BS. Hence, the received signal and SIR of $u_{0}\in\mathcal{U}_{2O\bar{C}}$ are
\begin{align}
&y_{2O\bar{C},0}=\frac{\mathbf{h}_{2,00}^{\dagger}\mathbf{f}_{2,0}x_{2,0}}{Y_{2}^{\frac{\alpha_{2}}{2}}}+\frac{\mathbf{h}_{1,10}^{\dagger}\mathbf{f}_{1,1}x_{1,1}}{Y_{1}^{\frac{\alpha_{1}}{2}}}+\sum_{\ell\in\Phi\left(\lambda_{1}\right)\backslash B_{1,0}}\frac{\mathbf{h}_{1,\ell0}^{\dagger}\mathbf{f}_{1,\ell}x_{1,\ell}}{\left|D_{1,\ell0}\right|^{\frac{\alpha_{1}}{2}}}+\sum_{\ell\in\Phi\left(\lambda_{2}\right)\backslash B_{2,0}}\frac{\mathbf{h}_{2,\ell0}^{\dagger}\mathbf{f}_{2,\ell}x_{2,\ell}}{\left|D_{2,\ell0}\right|^{\frac{\alpha_{2}}{2}}}\;,\label{eq:y0_ObarC}\\
&{\rm SIR}_{{\rm IN},2O\bar{C},0}= \frac{\frac{P_{2}}{Y_{2}^{\alpha_{2}}}\left|\mathbf{h}_{2,00}^{\dagger}\mathbf{f}_{2,0}\right|^{2}}{P_{2}\sum_{\ell\in\Phi\left(\lambda_{2}\right)\backslash B_{2,0}}\frac{\left|\mathbf{h}_{2,\ell0}^{\dagger}\mathbf{f}_{2,\ell}\right|^{2}}{\left|D_{2,\ell0}\right|^{\alpha_{2}}}+P_{1}\sum_{\ell\in\Phi\left(\lambda_{1}\right)\backslash B_{1,0}}\frac{\left|\mathbf{h}_{1,\ell0}^{\dagger}\mathbf{f}_{1,\ell}\right|^{2}}{\left|D_{1,\ell0}\right|^{\alpha_{1}}}+P_{1}\frac{\left|\mathbf{h}_{1,10}^{\dagger}\mathbf{f}_{1,1}\right|^{2}}{Y_{1}^{\alpha_{1}}}}\;.\label{eq:SIR_BbarC_v2}
\end{align} 
To facilitate the calculation of the rate coverage probability for $u_{0}\in\mathcal{U}_{2O\bar{C}}$ in Section \ref{sec:CPrate}, different from (\ref{eq:y0_pico}) and (\ref{eq:SINR_pico}),  we separate the dominant interferer (i.e., the nearest macro-BS) and the other interferers (i.e., the other macro-BSs) in the macro-cell tier to $u_{0}\in\mathcal{U}_{2O\bar{C}}$  in (\ref{eq:y0_ObarC}) and (\ref{eq:SIR_BbarC_v2}).

\begin{figure}[t]
\centering
\subfigure[$\:$]{
\includegraphics[width=3.1in]
{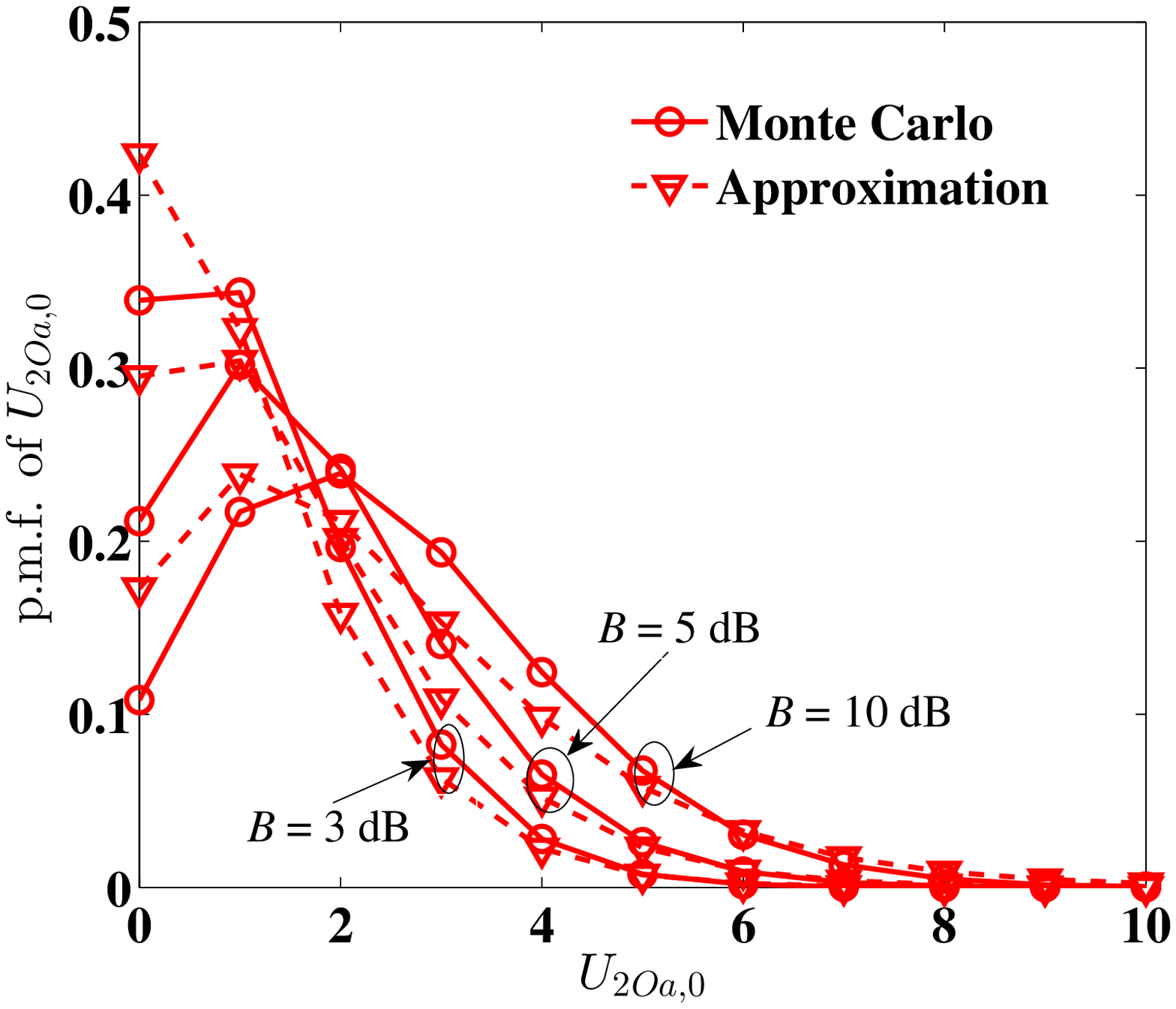}
\label{fig:p.m.f._UB0}
}
\subfigure[$\:$]{
\includegraphics[width=3.1in]{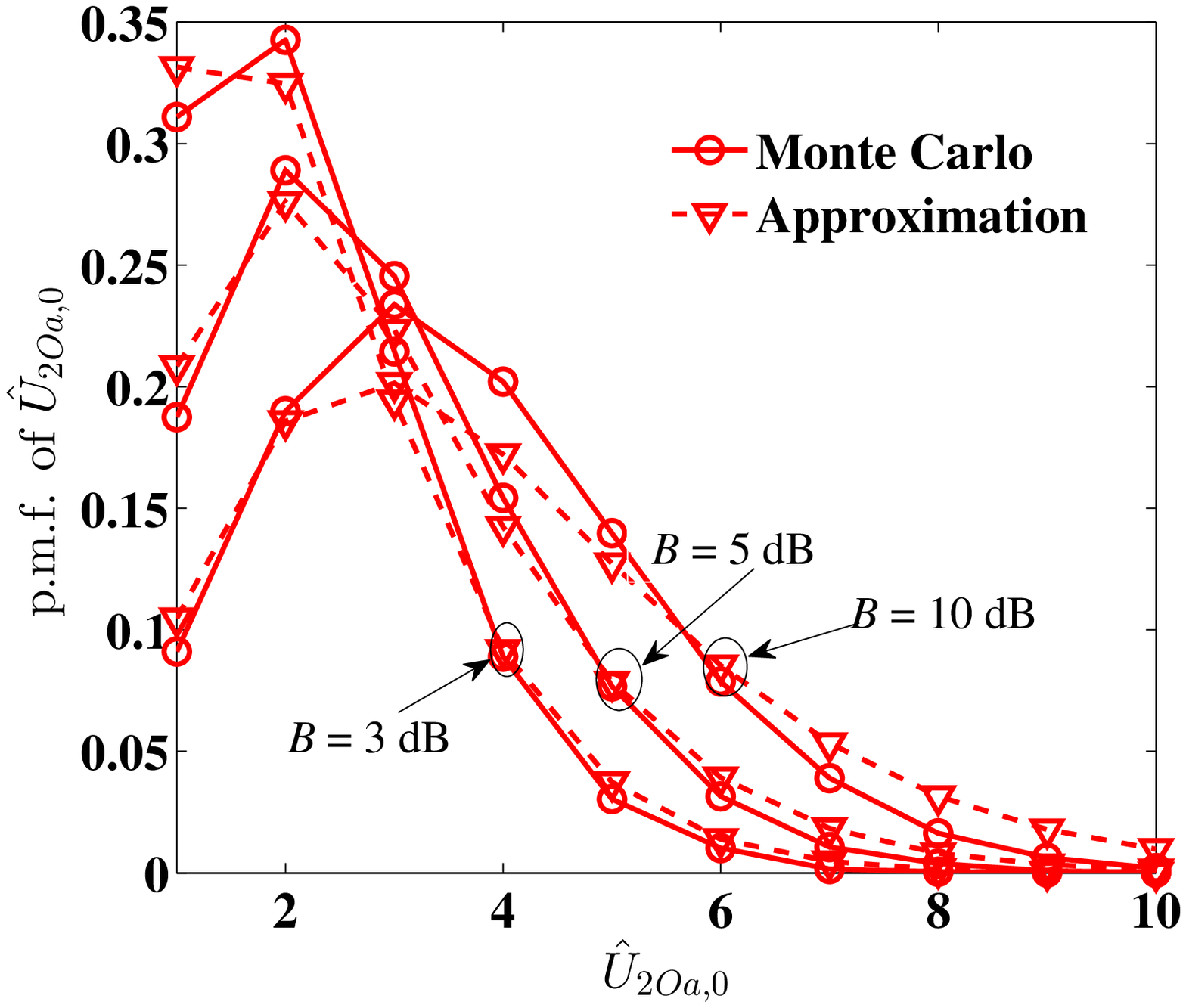}
\label{fig:p.m.f._hatUB0}
}
\caption{\scriptsize P.m.f. of $U_{2O_{a},0}$ and $\hat{U}_{2O_{a},0}$ for different bias factors $B$, at $\frac{P_{1}}{P_{2}}=20$ dB, $\alpha_{1}=\alpha_{2}=4$, $\lambda_{1}=0.0001$ nodes/m$^{2}$, and $\lambda_{2}=0.0005$ nodes/m$^{2}$.}\label{fig:CPrateVSeta_ABS_IN}
\vspace{-8mm}
\end{figure}

\subsection{Probability Mass Function of IN DoF and IN Probability}\label{subsec:preliminary}
\subsubsection{Probability Mass Function of IN DoF}
From (\ref{eq:SINR_marco}), we note that when $u_{0}\in\mathcal{U}_{1}$, the distribution of the effective channel gain ($\left|\mathbf{h}_{1,00}^{\dagger}\mathbf{f}_{1,0}\right|^{2}\dis {\rm Gamma}\left(N_{1}-u_{2OC,0},1\right)$) is related to the IN DoF at the typical user's serving macro-BS $u_{2OC,0}$. The probability mass function (p.m.f.) of $u_{2OC,0}$ is the basis of calculating the rate coverage probability in (\ref{eq:CPrate_def}). Let $U_{2O_{a},0}$ denote the number of active offloaded users of the typical user's serving macro-BS when $u_{0}\in\mathcal{U}_{1}$. In order to calculate the p.m.f. of $u_{2OC,0}$, we first calculate the p.m.f. of $U_{2O_{a},0}$. The p.m.f. of $U_{2O_{a},0}$ depends on the distributions of the number of active offloaded users in a fixed area and  the offloading area of the typical user's serving macro-BS,    
but its exact distribution is unknown. Similar to the approaches utilized in \cite{bai13,li14IC}, we approximate the distribution of the number of active offloaded users in a fixed area as a Poisson distribution. Moreover, we approximate the distribution of the offloading area using a linear-scaling-based approach proposed in \cite{singh13}. Based on these approximations, we calculate the p.m.f. of $U_{2O_{a},0}$ as follows:
\begin{lemma}\label{lem:p.m.f._offload_random}
When $u_{0}\in\mathcal{U}_{1}$, the p.m.f. of $U_{2O_{a},0}$ is approximated by
\begin{align}\label{eq:pmf_UB0}
{\rm Pr}\left(U_{2O_{a},0}=n\right) \approx \frac{3.5^{3.5}\Gamma\left(n+3.5\right)}{\Gamma(3.5)n!}\left(\frac{\lambda_{2}\mathcal{A}_{2O}}{\mathcal{A}_{2}\lambda_{1}}\right)^{n}\left(3.5+\frac{\lambda_{2}\mathcal{A}_{2O}}{\mathcal{A}_{2}\lambda_{1}}\right)^{-\left(n+3.5\right)}\;,\;n\ge 0
\end{align} where 
\begin{align}
\mathcal{A}_{2}&\define{\rm Pr}\left(u_{0}\in\mathcal{U}_{2}\right)=2\pi\lambda_{2}\int_{0}^{\infty}z\exp\left(-\pi\left(\lambda_{1}\left(\frac{P_{1}z^{\alpha_{2}}}{BP_{2}}\right)^{\frac{2}{\alpha_{1}}}+\lambda_{2}z^{2}\right)\right){\rm d}z\;\\
\mathcal{A}_{2O}&\define{\rm Pr}\left(u_{0}\in\mathcal{U}_{2O}\right)\notag\\
&=2\pi\lambda_{2}\int_{0}^{\infty}z\left(\exp\left(-\pi\lambda_{1}\left(\frac{P_{1}z^{\alpha_{2}}}{BP_{2}}\right)^{\frac{2}{\alpha_{1}}}\right)-\exp\left(-\pi\lambda_{1}\left(\frac{P_{1}z^{\alpha_{2}}}{P_{2}}\right)^{\frac{2}{\alpha_{1}}}\right)\right)\exp\left(-\pi\lambda_{2}z^{2}\right){\rm d}z\;.
\end{align}
\end{lemma}
\begin{proof}
See Appendix \ref{proof:p.m.f._offload_random}.
\end{proof}


Fig.\ \ref{fig:p.m.f._UB0} illustrates the accuracy of the p.m.f. approximation of $U_{2O_{a},0}$ in (\ref{eq:pmf_UB0}). We see that the p.m.f. approximation of $U_{2O_{a},0}$ is reasonably accurate for different bias factors.

Based on \emph{Lemma \ref{lem:p.m.f._offload_random}}, we can easily compute the p.m.f. of $u_{2OC,0}=\min\left(U,U_{2O_{a},0}\right)$ as follows:
\begin{lemma}\label{lem:pmf_u2OC}
When $u_{0}\in\mathcal{U}_{1}$, the p.m.f. of the IN DoF at the typical user's serving macro-BS is 
\begin{align}
{\rm Pr}\left(u_{2OC,0}=n\right)=
\begin{cases}
&{\rm Pr}\left(U_{2O_{a},0}=n\right),\hspace{1.15cm} {\rm for}\; 0\le n<U\\
&\sum_{u=n}^{\infty} {\rm Pr}\left(U_{2O_{a},0}=u\right), \quad {\rm for}\;\; n=U
\end{cases}
\;.
\end{align} 
\end{lemma}


\subsubsection{IN Probability}
As discussed in Section \ref{sub:IN_intro}, all the active offloaded users may not be simultaneously selected for IN. Let $\mathcal{E}_{2OC,0}(U)$ denote the event that $u_{0}$ is selected for IN in the IN scheme under design parameter $U$ given that $u_{0}\in\mathcal{U}_{2O}$. Here, ${\rm Pr}\left(\mathcal{E}_{2OC,0}(U)\right)$ is referred to as the IN probability and is the basis of calculating the rate coverage probability in (\ref{eq:CPrate_def}). Let $\hat{U}_{2O_{a},0}$ denote the number of active offloaded users that are offloaded from the typical user's nearest macro-BS when $u_{0}\in\mathcal{U}_{2O}$. 
To calculate ${\rm Pr}\left(\mathcal{E}_{2OC,0}(U)\right)$, we first calculate the p.m.f. of $\hat{U}_{2O_{a},0}$. 
Based on similar approximation approaches of deriving the p.m.f. of $U_{2O_{a},0}$ in \emph{Lemma \ref{lem:p.m.f._offload_random}}, we calculate the p.m.f. of $\hat{U}_{2O_{a},0}$ as follows: 
\begin{lemma}\label{lem:p.m.f._offload_tag}
When $u_{0}\in\mathcal{U}_{2O}$, the p.m.f. of $\hat{U}_{2O_{a},0}$ is approximated by
\begin{align}\label{eq:p.m.f._hatUB0}
{\rm Pr}\left(\hat{U}_{2O_{a},0}=n\right)\approx\frac{3.5^{3.5}\Gamma\left(n+3.5\right)}{\Gamma(n)\Gamma(3.5)}\left(\frac{\lambda_{2}\mathcal{A}_{2O}}{\mathcal{A}_{2}\lambda_{1}}\right)^{n-1}\left(3.5+\frac{\lambda_{2}\mathcal{A}_{2O}}{\mathcal{A}_{2}\lambda_{1}}\right)^{-\left(n+3.5\right)}\;,\;n\ge1\;.
\end{align} 
\end{lemma}
\begin{proof}
Similar to the proof of (\ref{eq:pmf_UB0}). The difference is that, in this proof, the distribution of the offloading area (where the offloaded users including $u_{0}$ may reside) of $u_{0}$'s \emph{nearest} macro-BS is used, instead of the distribution of the offloading area (where the offloaded users excluding $u_{0}$ may reside) of $u_{0}$'s \emph{serving} macro-BS (used in the proof of (\ref{eq:pmf_UB0})).  
\end{proof}



Fig.\ \ref{fig:p.m.f._hatUB0} illustrates the accuracy of the p.m.f. approximation of $\hat{U}_{2O_{a},0}$ in (\ref{eq:p.m.f._hatUB0}). We see that the p.m.f. approximation of $\hat{U}_{2O_{a},0}$ is reasonably accurate for different bias factors.


Based on \emph{Lemma \ref{lem:p.m.f._offload_tag}}, we can calculate the IN probability ${\rm Pr}\left(\mathcal{E}_{2OC,0}(U)\right)$ as follows: 
\begin{lemma}\label{lem:prob_IN}
When $u_{0}\in\mathcal{U}_{2O}$, the IN probability is 
\begin{align}\label{eq:prob_IN}
&{\rm Pr}\left(\mathcal{E}_{2OC,0}(U)\right)\\
=&U\left(\frac{\lambda_{1}\mathcal{A}_{2}}{\lambda_{2}\mathcal{A}_{2O}}\left(1-\left(1+\frac{\lambda_{2}\mathcal{A}_{2O}}{3.5\lambda_{1}\mathcal{A}_{2}}\right)^{-3.5}\right)-\sum_{n=1}^{U}\frac{1}{n}{\rm Pr}\left(\hat{U}_{2O_{a},0}=n\right)\right)+\sum_{n=1}^{U}{\rm Pr}\left(\hat{U}_{2O_{a},0}=n\right)\notag
\end{align} where ${\rm Pr}\left(\hat{U}_{2O_{a},0}=n\right)$ is given in (\ref{eq:p.m.f._hatUB0}).
\end{lemma}
\begin{proof}
According to total probability theorem, we have
\begin{align}\label{eq:IN_prob_U}
{\rm Pr}\left(\mathcal{E}_{2OC,0}(U)\right)=&\sum_{n=1}^{\infty}{\rm Pr}\left(\mathcal{E}_{2OC,0}(U)\Big|\hat{U}_{2O_{a},0}=n\right){\rm Pr}\left(\hat{U}_{2O_{a},0}=n\right)\notag\\
=&\sum_{n=1}^{U}{\rm Pr}\left(\hat{U}_{2O_{a},0}=n\right)+\sum_{n=U+1}^{\infty}\frac{U}{n}{\rm Pr}\left(\hat{U}_{2O_{a},0}=n\right)
\end{align} 
where $\sum_{n=1}^{\infty}\frac{1}{n}{\rm Pr}\left(\hat{U}_{2O_{a},0}=n\right)=\frac{\lambda_{1}\mathcal{A}_{2}}{\lambda_{2}\mathcal{A}_{2O}}\left(1-\left(1+\frac{\lambda_{2}\mathcal{A}_{2O}}{3.5\lambda_{1}\mathcal{A}_{2}}\right)^{-3.5}\right)$ is calculated by using a similar method as used in \cite[\emph{Proposition 2}]{yu13}.
\end{proof}

\section{Rate Coverage Probability Analysis of Interference Nulling}\label{sec:CPrate}
In this section, we investigate the rate coverage probability of the IN scheme. First, we derive the SIR coverage probability of each user type. Next, based on the SIR coverage probabilities of all user types, we obtain the rate coverage probability and its mean load approximation (MLA). 

\subsection{SIR Coverage Probability of Each User Type}
As discussed in Section \ref{sub:IN_intro}, under the IN scheme, the typical user $u_{0}$ can be in any user set $\mathcal{U}_{k}$, where $k\in\mathcal{K}\define\{1,2\bar{O},2OC,2O\bar{C}\}$. Let\footnote{Note that $\mathcal{S}_{{\rm IN},1}(\beta)$ is dependent of the design parameter $U$, while $\mathcal{S}_{{\rm IN},k}(\beta)$ is independent of $U$ for all $k\in\{2\bar{O},2OC,2O\bar{C}\}$. For notational simplicity, we do not make explicit the dependence of $\mathcal{S}_{{\rm IN},1}(\beta)$ on $U$. } $\mathcal{S}_{{\rm IN},k}(\beta)\define{\rm Pr}\left({\rm SIR}_{{\rm IN},k,0}>\beta|u_{0}\in\mathcal{U}_{k}\right)$ denote the SIR coverage probability of $u_{0}\in\mathcal{U}_{k}$ ($k\in\mathcal{K}$) under the IN scheme, where ${\rm SIR}_{{\rm IN},k,0}$ denotes the SIR of $u_{0}\in\mathcal{U}_{k}$ under the IN scheme and $\beta$ is the SIR threshold. Similar to (\ref{eq:CPrate_def}), the rate coverage probability of $u_{0}\in\mathcal{U}_{k}$ ($k\in\mathcal{K}$) under the IN scheme is defined as\footnote{Note that $\mathcal{R}_{{\rm IN},1}(\tau)$ is dependent of the design parameter $U$, while $\mathcal{R}_{{\rm IN},k}(\tau)$ is independent of $U$ for all $k\in\{2\bar{O},2OC,2O\bar{C}\}$. For notational simplicity, we do not make explicit the dependence of $\mathcal{R}_{{\rm IN},1}(\tau)$ on $U$. }
\begin{align}\label{eq:CPrate_each_IN}
\mathcal{R}_{{\rm IN},k}(\tau)&\define{\rm Pr}\left(R_{\rm IN,k,0}>\tau|u_{0}\in\mathcal{U}_{k}\right)\notag\\
&={\rm Pr}\left(\frac{W}{L_{0,j_{k}}}\log_{2}\left(1+{\rm SIR}_{{\rm IN},k,0}\right)>\tau|u_{0}\in\mathcal{U}_{k}\right)\notag\\
&={\rm E}_{L_{0,j_{k}}}\left[\mathcal{S}_{{\rm IN},k}\left(f\left(\frac{L_{0,j_{k}}\tau}{W}\right)\right)\right]
\end{align}
where $R_{\rm IN,k,0}$ denotes the rate of $u_{0}\in\mathcal{U}_{k}$ under the IN scheme, $f(x)=2^{x}-1$, and $L_{0,j_{k}}$ is the load of the typical user's serving BS which is in the $j_{k}$th tier. Here, $j_{k}$ is given in Table \ref{tab:para_B2larger}. According to (\ref{eq:CPrate_def}) and total probability theorem, the rate coverage probability of the IN scheme under design parameter $U$ can be written as 
\begin{align}\label{eq:CPrate_CP}
\mathcal{R}_{\rm IN}(U,\tau)&\define {\rm Pr}\left(R_{{\rm IN},0}>\tau\right)\notag\\
&=\sum_{k\in\mathcal{K}}\mathcal{A}_{k}{\rm E}_{L_{0,j_{k}}}\left[\mathcal{S}_{{\rm IN},k}\left(f\left(\frac{L_{0,j_{k}}\tau}{W}\right)\right)\right]
\end{align} 
where $R_{{\rm IN},0}$ is the rate of $u_{0}$ (which can be in any user set) under the IN scheme and $\mathcal{A}_{k}\define{\rm Pr}\left(u_{0}\in\mathcal{U}_{k}\right)$ ($k\in\mathcal{K}$). Specifically, $\mathcal{A}_{2OC}=\mathcal{A}_{2O}{\rm Pr}\left(\mathcal{E}_{2OC,0}(U)\right)$ and $\mathcal{A}_{2O\bar{C}}=\mathcal{A}_{2O}\left(1-{\rm Pr}\left(\mathcal{E}_{2OC,0}(U)\right)\right)$, where $\mathcal{A}_{2O}\define{\rm Pr}\left(u_{0}\in\mathcal{U}_{2O}\right)$ is given in \emph{Lemma \ref{lem:p.m.f._offload_random}}. Note that $\mathcal{A}_{k}$ ($k\in\{1,2\bar{O},2O\}$) is independent of $U$. 
In this part, we calculate $\mathcal{S}_{{\rm IN},k}(\beta)$. Based on $\mathcal{S}_{{\rm IN},k}(\beta)$, we shall calculate $\mathcal{R}_{\rm IN}(U,\tau)$ in the next part. Let $R_{jk}$ denote the minimum possible distance between $u_0\in\mathcal{U}_{k}$ ($k\in\mathcal{K}$) and its nearest interferer in the $j$th tier ($j=1,2$). Note that $\mathcal{S}_{{\rm IN},k}(\beta)={\rm E}_{R_{1k},R_{2k}}\left[\mathcal{S}_{{\rm IN},k,R_{1k},R_{2k}}(r_{1k},r_{2k},\beta)\right]$, where $\mathcal{S}_{{\rm IN},k,R_{1k},R_{2k}}(r_{1k},r_{2k},\beta)\define {\rm Pr}\Big({\rm SIR}_{{\rm IN},k,0}>\beta|u_{0}\in\mathcal{U}_{k},R_{1k}=r_{1k},R_{2k}=r_{2k}\Big)$ denotes the conditional SIR coverage probability\footnote{When $u_{0}\in\mathcal{U}_{1}$, we also condition on $u_{2OC,0}$. For notational simplicity, we do not make this dependence explicit.}. To calculate $\mathcal{S}_{{\rm IN},k}(\beta)$, we first need to calculate $\mathcal{S}_{{\rm IN},k,R_{1k},R_{2k}}(r_{1k},r_{2k},\beta)$, which is provided as follows:

\begin{lemma}\label{theo:CP_condi}
The conditional SIR coverage probability of $u_{0}\in\mathcal{U}_{k}$ under the IN scheme is 
\begin{align}\label{eq:condiCP}
\mathcal{S}_{{\rm IN},k,R_{1k},R_{2k}}(r_{1k},r_{2k},\beta)=\sum_{n=0}^{M_{k}-1}\mathcal{T}_{k,R_{1k},R_{2k}}\left(n,r_{1k},r_{2k},\beta\right) 
\end{align} 
where $k\in\mathcal{K}$ and 
\begin{align}\label{eq:Tn}
&\mathcal{T}_{k,R_{1k},R_{2k}}\left(n,r_{1k},r_{2k},\beta\right)\notag\\
=&
\begin{cases}
\frac{1}{n!}\sum_{n_{1}=0}^{n}\binom{n}{n_{1}}\tilde{\mathcal{L}}_{I_{1}}^{(n_{1})}\left(s,r_{1k}\right)\Big|_{s=\beta Y_{j_{k}}^{\alpha_{j_{k}}}\frac{P_{1}}{P_{j_{k}}}}\tilde{\mathcal{L}}_{I_{2}}^{(n-n_{1})}\left(s,r_{2k}\right)\Big|_{s=\beta Y_{j_{k}}^{\alpha_{j_{k}}}\frac{P_{2}}{P_{j_{k}}}}\;,\hspace{0mm}{\rm if}\hspace{2mm} k\in\{1,2\bar{O},2OC\}\\
\frac{1}{n!}\sum_{(q_{a})_{a=1}^{3}\in\mathcal{Q}_{3}}\binom{n}{q_{1},q_{2},q_{3}}\tilde{\mathcal{L}}_{I_{1}}^{(q_{1})}\left(s,r_{1k}\right)\Big|_{s=\beta Y_{j_{k}}^{\alpha_{j_{k}}}\frac{P_{1}}{P_{j_{k}}}}\tilde{\mathcal{L}}_{I_{2}}^{(q_{2})}\left(s,r_{2k}\right)\Big|_{s=\beta Y_{j_{k}}^{\alpha_{j_{k}}}\frac{P_{2}}{P_{j_{k}}}}\Gamma\left(q_{3}+1\right)\\
\hspace{44mm}\times\left(\beta\frac{P_{1}Y_{j_{k}}^{\alpha_{j_{k}}}}{P_{j_{k}}r_{1k}^{\alpha_{1}}}\right)^{q_{3}}\left(1+\beta\frac{P_{1}Y_{j_{k}}^{\alpha_{j_{k}}}}{P_{j_{k}}r_{1k}^{\alpha_{1}}}\right)^{-\left(q_{3}+1\right)}\;,\hspace{0mm}{\rm if}\hspace{2mm} k=2O\bar{C}
\end{cases}
.
\end{align} 
Here, $j_{k}$, $r_{1k}$, $r_{2k}$, and $M_{k}$ are given in Table \ref{tab:para_B2larger}, 
\begin{align}
\mathcal{Q}_{3}&\define\{(q_{a})_{a=1}^{3}|q_{a}\in\mathbb{N}^{0},\sum_{a=1}^{3}q_{a}=n\}\;,\\
\tilde{\mathcal{L}}_{I_j}^{(m)}\left(s,r_{jk}\right)&=\mathcal{L}_{I_j}\left(s,r_{jk}\right)\notag\\&\hspace{-16mm}\times\sum_{(p_{a})_{a=1}^{m}\in\mathcal{M}_{m}}\frac{m!}{\prod_{a=1}^{m}p_{a}!}\prod_{a=1}^{m}\left(\frac{2\pi}{\alpha_{j}}\lambda_{j}s^{\frac{2}{\alpha_{j}}}B^{'}\left(1+\frac{2}{\alpha_{j}},a-\frac{2}{\alpha_{j}},\frac{1}{1+sr_{jk}^{-\alpha_{j}}}\right)\right)^{p_{a}},\;{\rm for}\;\;j\in\{1,2\}
\end{align} 
where\footnote{$\mathcal{L}_{I_j}\left(s,r_{jk}\right)$ is the Laplace transform of the aggregated interference $I_{j}=\sum_{\ell\in\Phi\left(\lambda_{j}\right)\backslash B(0,r_{jk})}\frac{\left|\mathbf{h}_{j,\ell0}^{\dagger}\mathbf{f}_{j,\ell}\right|^{2}}{\left|D_{j,\ell0}\right|^{\alpha_{j}}}$ from the $j$th tier.} 
\begin{align}
\mathcal{L}_{I_j}\left(s,r_{jk}\right)&=\exp\left(-\frac{2\pi }{\alpha_{j}}\lambda_{j}s^{\frac{2}{\alpha_{j}}}B^{'}\left(\frac{2}{\alpha_{j}},1-\frac{2}{\alpha_{j}},\frac{1}{1+sr_{jk}^{-\alpha_{j}}}\right)\right)\;, \\
\mathcal{M}_{m}&\define \Big\{(p_{a})_{a=1}^{m}|p_{a}\in\mathbb{N}^{0},\sum_{a=1}^{m}a\cdot p_{a}=m\Big\}\;,
\end{align} 
and $B^{'}(a,b,z) \define \int_{z}^{1}u^{a-1}(1-u)^{b-1}{\rm d}u$ $(0<z<1)$ is the complementary incomplete Beta function \cite{gupta13}.
\end{lemma}
\begin{proof}
See Appendix \ref{proof:theo_CPcondi}.
\end{proof}

Note that $\mathcal{T}_{k,R_{1k},R_{2k}}\left(n,r_{1k},r_{2k},\beta\right)$ in (\ref{eq:condiCP}) can be interpreted as the gain of the SIR coverage probability when the DoF for boosting the desired signal to $u_{0}\in\mathcal{U}_{k}$ at its serving BS is changed from $n$ to $n+1$. 

\begin{table}[t]
\caption{Parameter values under the IN scheme when $u_{0}\in\mathcal{U}_{k}$ with $k\in\mathcal{K}$}\label{tab:para_B2larger}
\begin{center}
\vspace{-6mm}
\begin{tabular}{|c|c|c|c|c|c|}
\hline
$k$&$j_{k}$&$r_{1k}$&$r_{2k}$&$M_{k}$\\
\hline
$1$ &$1$&$Y_{1}$ &$\left(\frac{P_{2}B}{P_{1}}\right)^{\frac{1}{\alpha_{2}}}Y_{1}^{\frac{\alpha_{1}}{\alpha_{2}}}$&$N_{1}-u_{2O,0}$\\
\hline
$2\bar{O}$&$2$ & $\left(\frac{P_{1}}{P_{2}}\right)^{\frac{1}{\alpha_{1}}}Y_{2}^{\frac{\alpha_{2}}{\alpha_{1}}}$&$Y_{2}$&$N_{2}$\\
\hline
$2OC$ & $2$ & $Y_{1}$ & $Y_{2}$ & $N_{2}$\\
\hline
$2O\bar{C}$ & $2$ & $Y_{1}$ & $Y_{2}$ & $N_{2}$\\
\hline
\end{tabular}
\end{center}
\vspace{-6mm}
\end{table}

Based on \emph{Lemma \ref{lem:pmf_u2OC}} and \emph{Lemma \ref{theo:CP_condi}}, we have the SIR coverage probability $\mathcal{S}_{{\rm IN},k}(\beta)$ of $u_{0}\in\mathcal{U}_{k}$ ($k\in\mathcal{K}$) as follows:
\begin{theorem}\label{cor:CP_uncondi}
The SIR coverage probability of $u_{0}\in\mathcal{U}_{k}$ under the IN scheme is 
\begin{align}
\mathcal{S}_{{\rm IN},1}(\beta) &= \sum_{n=0}^{U}\left(\int_{0}^{\infty}\mathcal{S}_{{\rm IN},1,Y_{1}}(y,\beta)f_{Y_{1}}(y){\rm d}y\right){\rm Pr}\left(u_{2OC,0}=n\right)\;,\label{eq:CP1}\\
\mathcal{S}_{{\rm IN},2\bar{O}}(\beta)&=\int_{0}^{\infty}\mathcal{S}_{{\rm IN},2\bar{O},Y_{2}}(y,\beta)f_{Y_{2}}(y){\rm d}y\;,\label{eq:CP2barO}\\
\mathcal{S}_{{\rm IN},2OC}(\beta)&=\int_{0}^{\infty}\int_{\left(\frac{P_{2}}{P_{1}}\right)^{\frac{1}{\alpha_{2}}}x^{\frac{\alpha_{1}}{\alpha_{2}}}}^{\left(\frac{BP_{2}}{P_{1}}\right)^{\frac{1}{\alpha_{2}}}x^{\frac{\alpha_{1}}{\alpha_{2}}}}\mathcal{S}_{{\rm IN},2OC,Y_{1},Y_{2}}(x,y,\beta)f_{Y_{1},Y_{2}}(x,y){\rm d}y{\rm d}x\;,\label{eq:CP2OC}\\
\mathcal{S}_{{\rm IN},2O\bar{C}}(\beta)&=\int_{0}^{\infty}\int_{\left(\frac{P_{2}}{P_{1}}\right)^{\frac{1}{\alpha_{2}}}x^{\frac{\alpha_{1}}{\alpha_{2}}}}^{\left(\frac{BP_{2}}{P_{1}}\right)^{\frac{1}{\alpha_{2}}}x^{\frac{\alpha_{1}}{\alpha_{2}}}}\mathcal{S}_{{\rm IN},2O\bar{C},Y_{1},Y_{2}}(x,y,\beta)f_{Y_{1},Y_{2}}(x,y){\rm d}y{\rm d}x\;,\label{eq:CP2ObarC}
\end{align} 
where 
\begin{align}
f_{Y_{1}}(y)&= \frac{2\pi\lambda_{1}}{\mathcal{A}_{1}}y\exp\left(-\pi\left(\lambda_{1}y^{2}+\lambda_{2}\left(\frac{P_{2}B}{P_{1}}\right)^{\frac{2}{\alpha_{2}}}y^{\frac{2\alpha_{1}}{\alpha_{2}}}\right)\right)\;,\\
f_{Y_{2}}(y)&=\frac{2\pi\lambda_{2}}{\mathcal{A}_{2\bar{O}}}y\exp(-\pi\lambda_{2}y^{2})\exp\left(-\pi\lambda_{1}\left(\frac{P_{1}}{P_{2}}\right)^{\frac{2}{\alpha_{1}}}y^{\frac{2\alpha_{2}}{\alpha_{1}}}\right)\;,\\
f_{Y_{1},Y_{2}}(x,y)&=\frac{4\pi^{2}\lambda_{1}\lambda_{2}}{\mathcal{A}_{2O}}xy\exp\left(-\pi\left(\lambda_{1}x^{2}+\lambda_{2}y^{2}\right)\right)\;.
\end{align} 
Here, $\mathcal{A}_{2O}$ is given in \emph{Lemma \ref{lem:p.m.f._offload_random}}, and 
\begin{align}
\mathcal{A}_{1}&=2\pi\lambda_{1}\int_{0}^{\infty}z\exp\left(-\pi\lambda_{1}z^{2}\right)\exp\left(-\pi\left(\lambda_{2}\left(\frac{BP_{2}}{P_{1}}\right)^{\frac{2}{\alpha_{2}}}z^{\frac{2\alpha_{1}}{\alpha_{2}}}\right)\right){\rm d}z\;,\\ 
\mathcal{A}_{2\bar{O}}&=2\pi\lambda_{2}\int_{0}^{\infty}z\exp\left(-\pi\lambda_{1}\left(\frac{P_{1}}{P_{2}}\right)^{\frac{2}{\alpha_{1}}}z^{\frac{2\alpha_{2}}{\alpha_{1}}}-\pi\lambda_{2}z^{2}\right){\rm d}z\;.
\end{align}
\end{theorem}
\begin{proof}
Follows by removing the conditions of $\mathcal{S}_{{\rm IN},k,R_{1k},R_{2k}}(r_{1k},r_{2k},\beta)$ on $R_{jk}$ ($j=1,2$) in (\ref{eq:condiCP}). Here, $f_{Y_{1}}(y)$, $f_{Y_{2}}(y)$, $\mathcal{A}_{1}$ and $\mathcal{A}_{2\bar{O}}$ are given in \cite{singh13}, and $f_{Y_{1},Y_{2}}(x,y)$ is given in \cite{sakr14}. 
\end{proof}


\subsection{Rate Coverage Probability}\label{subsec:CPrate_anal}
Based on the SIR coverage probability of $u_{0}\in\mathcal{U}_{k}$ $(k\in\mathcal{K})$ in \emph{Theorem \ref{cor:CP_uncondi}} and the connection between $\mathcal{S}_{{\rm IN},k}(\beta)$ and $\mathcal{R}_{\rm IN}(U,\tau)$ in (\ref{eq:CPrate_CP}), we have the rate coverage probability as follows:
\begin{theorem}\label{cor:CPrate_overall}
The rate coverage probability of the IN scheme under $U$ is
\begin{align}\label{eq:CPrate_IN_overall}
\mathcal{R}_{\rm IN}(U,\tau)=&\mathcal{A}_{1}\mathcal{R}_{{\rm IN},1}(\tau)+\mathcal{A}_{2\bar{O}}\mathcal{R}_{{\rm IN},2\bar{O}}(\tau)\notag\\
&\hspace{3mm}+\mathcal{A}_{2O}{\rm Pr}\left(\mathcal{E}_{2OC,0}(U)\right)\mathcal{R}_{{\rm IN},2OC}(\tau)+\mathcal{A}_{2O}\left(1-{\rm Pr}\left(\mathcal{E}_{2OC,0}(U)\right)\right)\mathcal{R}_{{\rm IN},2O\bar{C}}(\tau)
\end{align} 
where $\mathcal{A}_{1}$, $\mathcal{A}_{2\bar{O}}$ and $\mathcal{A}_{2O}$ are given in \emph{Theorem \ref{cor:CP_uncondi}}, and   
\begin{align}
&\mathcal{R}_{{\rm IN},1}(\tau) = \sum_{n\ge 1}{\rm Pr}\left(L_{0,1}=n\right)\mathcal{S}_{{\rm IN}, 1}\left(f\left(\frac{n\tau}{W}\right)\right)\;,\label{eq:CPrate1}\\
&\mathcal{R}_{{\rm IN},2\bar{O}}(\tau)=\sum_{n\ge1}{\rm Pr}\left(L_{0,2}=n\right)\mathcal{S}_{{\rm IN},2\bar{O}}\left(f\left(\frac{n\tau}{W}\right)\right)\;,\label{eq:CPrate2barO}\\ 
&\mathcal{R}_{{\rm IN},2OC}(\tau)=\sum_{n\ge1}{\rm Pr}(L_{0,2}=n)\mathcal{S}_{{\rm IN},2OC}\left(f\left(\frac{n\tau}{W}\right)\right)\;,\label{eq:CPrate2OC}\\
&\mathcal{R}_{{\rm IN},2O\bar{C}}(\tau)=\sum_{n\ge1}{\rm Pr}(L_{0,2}=n)\mathcal{S}_{{\rm IN},2O\bar{C}}\left(f\left(\frac{n\tau}{W}\right)\right)\;.\label{eq:CPrate2ObarC}
\end{align} 
Here, $\mathcal{S}_{{\rm IN},k}(\cdot)$ is given by (\ref{eq:CP1})--(\ref{eq:CP2ObarC}), ${\rm Pr}\left(L_{0,1}= n\right) = \frac{3.5^{3.5}\Gamma\left(n+3.5\right)\left(\frac{\lambda_{u}\mathcal{A}_{1}}{\lambda_{1}}\right)^{n-1}}{\Gamma(3.5)(n-1)!\left(\frac{\lambda_{u}\mathcal{A}_{1}}{\lambda_{1}}+3.5\right)^{n+3.5}}$, and ${\rm Pr}\left(L_{0,2}= n\right)$\\$ = \frac{3.5^{3.5}\Gamma\left(n+3.5\right)\left(\frac{\lambda_{u}\mathcal{A}_{2}}{\lambda_{2}}\right)^{n-1}}{\Gamma(3.5)(n-1)!\left(\frac{\lambda_{u}\mathcal{A}_{2}}{\lambda_{2}}+3.5\right)^{n+3.5}}$.
\end{theorem}
\begin{proof}
Follows by conditioning on the load (i.e., $L_{1,0}$ or $L_{2,0}$), calculating the conditional rate coverage probability according to \emph{Lemma \ref{theo:CP_condi}}, and  removing the conditions on the load (i.e., $L_{0,1}$ or $L_{0,2}$). Note that the p.m.f. of $L_{0,1}$ is given in \cite[\emph{Lemma 3}]{singh13}, and the p.m.f. of $L_{0,2}$ can be calculated using a similar approach.
\end{proof}

Note that the expression of $\mathcal{R}_{\rm IN}(U,\tau)$ in (\ref{eq:CPrate_IN_overall}) of \emph{Theorem \ref{cor:CPrate_overall}} is difficult to compute and analyze due to the infinite summations over $n$ in (\ref{eq:CPrate1})--(\ref{eq:CPrate2ObarC}). To simplify the expression of $\mathcal{R}_{\rm IN}(U,\tau)$ in (\ref{eq:CPrate_IN_overall}), we use the mean of the random load (i.e., ${\rm E}\left[L_{0,j}\right]$) to approximate the random load (i.e., $L_{0,j}$), where $j=1,2$ \cite{singh13,singh13may}. The simplification is achieved due to the elimination of the infinite summation over $n$. In other words, by replacing $L_{0,j}$ with ${\rm E}\left[L_{0,j}\right]$ in (\ref{eq:CPrate_CP}), we can obtain the rate coverage probability with MLA of the IN scheme under $U$, denoted as $\mathcal{\bar{R}}_{\rm IN}(U,\tau)$, as follows:
\begin{corollary}\label{cor:CPrate_MLA}
The rate coverage probability with MLA of the IN scheme under $U$ is  
\begin{align}\label{eq:CPrateMLA_IN}
\mathcal{\bar{R}}_{\rm IN}(U,\tau)=&\mathcal{A}_{1}\mathcal{\bar{R}}_{{\rm IN},1}(\tau)+\mathcal{A}_{2\bar{O}}\mathcal{\bar{R}}_{{\rm IN},2\bar{O}}(\tau)\notag\\
&\hspace{3mm}+\mathcal{A}_{2O}{\rm Pr}\left(\mathcal{E}_{2OC,0}(U)\right)\mathcal{\bar{R}}_{{\rm IN},2OC}(\tau)+\mathcal{A}_{2O}\left(1-{\rm Pr}\left(\mathcal{E}_{2OC,0}(U)\right)\right)\mathcal{\bar{R}}_{{\rm IN},2O\bar{C}}(\tau)
\end{align} 
where 
\begin{align}
\mathcal{\bar{R}}_{{\rm IN},1}(\tau)&=\mathcal{S}_{{\rm IN},1}\left(f\left(\frac{{\rm E}\left[L_{0,1}\right]\tau}{W}\right)\right)\;,\hspace{2mm}\mathcal{\bar{R}}_{{\rm IN},2\bar{O}}(\tau)=\mathcal{S}_{{\rm IN},2\bar{O}}\left(f\left(\frac{{\rm E}\left[L_{0,2}\right]\tau}{W}\right)\right)\;, \\
\mathcal{\bar{R}}_{{\rm IN},2OC}(\tau)&=\mathcal{S}_{{\rm IN},2OC}\left(f\left(\frac{{\rm E}\left[L_{0,2}\right]\tau}{W}\right)\right)\;,\hspace{2mm} \mathcal{\bar{R}}_{{\rm IN},2O\bar{C}}(\tau)=\mathcal{S}_{{\rm IN},2O\bar{C}}\left(f\left(\frac{{\rm E}\left[L_{0,2}\right]\tau}{W}\right)\right)
\end{align}
with $\mathcal{S}_{{\rm IN},k}(\cdot)$ given by (\ref{eq:CP1})--(\ref{eq:CP2ObarC}), ${\rm E}\left[L_{0,1}\right]=1+1.28\frac{\lambda_{u}\mathcal{A}_{1}}{\lambda_{1}}$, and ${\rm E}\left[L_{0,2}\right]=1+1.28\frac{\lambda_{u}\mathcal{A}_{2}}{\lambda_{2}}$. Here, $\mathcal{A}_{1}$, $\mathcal{A}_{2\bar{O}}$ and $\mathcal{A}_{2O}$ are given in \emph{Theorem \ref{cor:CP_uncondi}}, and $\mathcal{A}_{2}$ is given in \emph{Lemma \ref{lem:p.m.f._offload_random}}.
\end{corollary}
\begin{proof}
Follows by replacing $L_{0,j}$ in (\ref{eq:CPrate_IN_overall}) with ${\rm E}\left[L_{0,j}\right]$, where $j=1,2$. Note that ${\rm E}\left[L_{0,1}\right]$ is given in \cite{singh13}, and ${\rm E}\left[L_{0,2}\right]$ can be calculated using a similar approach.  
\end{proof}

\begin{figure}[t] \centering
\includegraphics[width=4.2in]{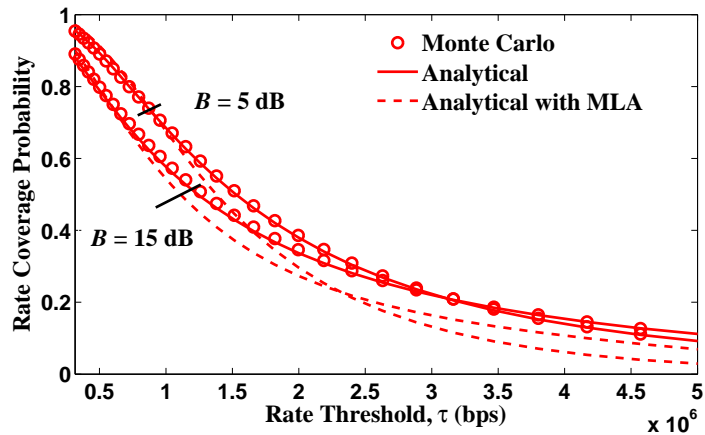}
\caption{\scriptsize Rate coverage probability vs. rate threshold $\tau$ for different bias factors $B$, at $\alpha_{1}=\alpha_{2}=4$, $\frac{P_{1}}{P_{2}}=10$ dB, $N_{1}=8$, $N_{2}=4$, $U=4$, $W=10$ MHz, $\lambda_{1}=0.0001$ nodes/m$^{2}$, and $\lambda_{2}=0.0005$ nodes/m$^{2}$.}\label{fig:CPrate_diffB}
\vspace{-8mm}
\end{figure}

Fig.\ \ref{fig:CPrate_diffB} plots the rate coverage probability of the IN scheme  vs. rate threshold $\tau$ for different bias factors $B$. We see from Fig.\ \ref{fig:CPrate_diffB} that the `Analytical' curves (i.e., $\mathcal{R}_{\rm IN}(U,\tau)$ in \emph{Theorem \ref{cor:CPrate_overall}}) closely match with the `Monte Carlo' curves, although $\mathcal{R}_{\rm IN}(U,\tau)$ is derived based on some approximations, as illustrated in Section \ref{subsec:preliminary}. Moreover, we observe that the `Analytical with MLA' curves (i.e., $\mathcal{\bar{R}}_{\rm IN}(U,\tau)$ in \emph{Corollary \ref{cor:CPrate_MLA}}) are close to the `Analytical' curves (i.e., $\mathcal{R}_{\rm IN}(U,\tau)$ in \emph{Theorem \ref{cor:CPrate_overall}}), especially when $\tau$ is not very large. Hence, for analytical tractability, we will investigate the rate coverage probability with MLA $\mathcal{\bar{R}}_{\rm IN}(U,\tau)$ in the remaining part of this paper.

\section{Rate Coverage Probability Optimization of Interference Nulling} 
In this section, we consider the rate coverage probability optimization of the IN scheme. For a fixed bias factor $B$, the optimal design parameter $U^{*}(\tau)$, which maximizes the (overall) rate coverage probability $\mathcal{\bar{R}}_{\rm IN}(U,\tau)$, is defined as follows:
\begin{align}\label{eq:optU_def}
U^{*}(\tau)\stackrel{\Delta}{=} {\rm arg}\; \max_{U\in\{0,1,\ldots,N_1-1\}} \mathcal{\bar{R}}_{\rm IN}(U,\tau)\;.
\end{align} 
Note that (\ref{eq:optU_def}) is an integer programming problem with a very complicated objective function $\mathcal{\bar{R}}_{\rm IN}(U,\tau)$. It is thus difficult to obtain the closed-form optimal solution $U^{*}(\tau)$ to the problem in (\ref{eq:optU_def}). To address this challenge, in the following, we first characterize the rate coverage probability change when the design parameter is changed from $U-1$ to $U$. Then, based on it, we study some properties of $U^{*}(\tau)$ for small and general rate threshold regimes, respectively.

\subsection{Rate Coverage Probability Change}
First, we define $\Delta\mathcal{\bar{R}}_{\rm IN}(U,\tau)\define \mathcal{\bar{R}}_{\rm IN}(U,\tau)-\mathcal{\bar{R}}_{\rm IN}(U-1,\tau)$ as the change of $\mathcal{\bar{R}}_{\rm IN}(U,\tau)$ when the design parameter is changed from $U-1$ to $U$, where $U\in\{1,\ldots,N_{1}-1\}$. By (\ref{eq:CPrateMLA_IN}), $\Delta\mathcal{\bar{R}}_{\rm IN}(U,\tau)$ can be decomposed into three parts as follows:
\begin{align}\label{eq:CP_change}
\Delta\mathcal{\bar{R}}_{\rm IN}(U,\tau)=\mathcal{A}_{1}\Delta\mathcal{\bar{R}}_{{\rm IN},1}(U,\tau)+\mathcal{A}_{2\bar{O}}\Delta\mathcal{\bar{R}}_{{\rm IN},2\bar{O}}(\tau)+\mathcal{A}_{2O}\Delta\mathcal{\bar{R}}_{{\rm IN},2O}(U,\tau)
\end{align} where\footnote{From now on, we make explicit the dependence of $\mathcal{\bar{R}}_{{\rm IN},1}(\tau)$ on $U$.} $\Delta\mathcal{\bar{R}}_{{\rm IN},1}(U,\tau)\define\mathcal{\bar{R}}_{{\rm IN},1}(U,\tau)-\mathcal{\bar{R}}_{{\rm IN},1}(U-1,\tau)$ denotes the rate coverage probability change of a macro-user, $\Delta\mathcal{\bar{R}}_{{\rm IN},2\bar{O}}(\tau)\define\mathcal{\bar{R}}_{{\rm IN},2\bar{O}}\left(\tau\right)-\mathcal{\bar{R}}_{{\rm IN},2\bar{O}}\left(\tau\right)$ denotes the rate coverage probability change of an unoffloaded pico-user, and 
\begin{align}\label{eq:CPloss_macro}
\Delta\mathcal{\bar{R}}_{{\rm IN},2O}(U,\tau)\define&\mathcal{\bar{R}}_{{\rm IN},2O}(U,\tau)-\mathcal{\bar{R}}_{{\rm IN},2O}(U-1,\tau)\notag\\
=&\left({\rm Pr}\left(\mathcal{E}_{2OC,0}(U)\right)-{\rm Pr}\left(\mathcal{E}_{2OC,0}(U-1)\right)\right)\left(\mathcal{\bar{R}}_{{\rm IN},2OC}(\tau)-\mathcal{\bar{R}}_{{\rm IN},2O\bar{C}}(\tau)\right)
\end{align} 
denotes the rate coverage probability change of an offloaded user. Here, $\mathcal{\bar{R}}_{{\rm IN},2O}(U,\tau)\define{\rm Pr}\left(\mathcal{E}_{2OC,0}(U)\right)$\\$\times\mathcal{\bar{R}}_{{\rm IN},2OC}(\tau)+\left(1-{\rm Pr}\left(\mathcal{E}_{2OC,0}(U)\right)\right)\mathcal{\bar{R}}_{{\rm IN},2O\bar{C}}(\tau)$ denotes the rate coverage probability of an offloaded user. 

Next, we analyze $\Delta\mathcal{\bar{R}}_{{\rm IN},1}(U,\tau)$, $\Delta\mathcal{\bar{R}}_{{\rm IN},2\bar{O}}(\tau)$ and $\Delta\mathcal{\bar{R}}_{{\rm IN},2O}(U,\tau)$ in the following lemma:
\begin{lemma}\label{lem:DeltaS_sign}
i) $\Delta\mathcal{\bar{R}}_{{\rm IN},1}(U,\tau)<0$, ii) $\Delta\mathcal{\bar{R}}_{{\rm IN},2\bar{O}}(\tau)=0$, and iii) $\Delta\mathcal{\bar{R}}_{{\rm IN},2O}(U,\tau)>0$.
\end{lemma}
\begin{proof}
See Appendix \ref{proof:DeltaS_sign}. 
\end{proof}


Based on \emph{Lemma \ref{lem:DeltaS_sign}}, $\Delta\mathcal{\bar{R}}_{\rm IN}(U,\tau)$ can be simplified as follows:   
\begin{align}\label{eq:CPdelta_diffU}
\Delta\mathcal{\bar{R}}_{\rm IN}(U,\tau)=\mathcal{A}_{2O}\Delta\mathcal{\bar{R}}_{{\rm IN},2O}(U,\tau)-\mathcal{A}_{1}\left|\Delta\mathcal{\bar{R}}_{{\rm IN},1}(U,\tau)\right|\;.
\end{align} 
where $\mathcal{A}_{2O}\Delta\mathcal{\bar{R}}_{{\rm IN},2O}(U,\tau)$ and $\mathcal{A}_{1}\left|\Delta\mathcal{\bar{R}}_{{\rm IN},1}(U,\tau)\right|$ are referred to as the ``gain" and the ``penalty" of the IN scheme, respectively. Whether $\Delta\mathcal{\bar{R}}_{\rm IN}(U,\tau)$ is positive or not depends on whether the ``gain" dominates the ``penalty" or not. Therefore, to maximize $\mathcal{\bar{R}}_{\rm IN}(U,\tau)$, we can study the properties of $\Delta\mathcal{\bar{R}}_{\rm IN}(U,\tau)$ in (\ref{eq:CPdelta_diffU}) w.r.t.\ $U$ by comparing $\mathcal{A}_{2O}\Delta\mathcal{\bar{R}}_{{\rm IN},2O}(U,\tau)$ and $\mathcal{A}_{1}\left|\Delta\mathcal{\bar{R}}_{{\rm IN},1}(U,\tau)\right|$.


\subsection{Rate Coverage Probability Optimization When $\tau\to0$}
In this part, we obtain $U^{*}(\tau)$ when $\tau\to0$ by comparing $\mathcal{A}_{2O}\Delta\mathcal{\bar{R}}_{{\rm IN},2O}(U,\tau)$ and $\mathcal{A}_{1}\left|\Delta\mathcal{\bar{R}}_{{\rm IN},1}(U,\tau)\right|$. First, we characterize $\left|\mathcal{\bar{R}}_{{\rm IN},1}(U,\tau)\right|$ and $\mathcal{\bar{R}}_{{\rm IN},2O}(U,\tau)$. To characterize $\left|\Delta\mathcal{\bar{R}}_{{\rm IN},1}(U,\tau)\right|$, by \emph{Corollary \ref{cor:CPrate_MLA}}, \emph{Theorem \ref{cor:CP_uncondi}}, and \emph{Lemma \ref{theo:CP_condi}}, we first characterize $\mathcal{T}_{k,R_{1k},R_{2k}}\left(n,r_{1k},r_{2k},2^{{\rm E}\left[L_{0,j_{k}}\right]\tau/W}-1\right)$, which indicates the SIR coverage probability gain of $u_{0}\in\mathcal{U}_{k}$ achieved when the DoF for boosting the desired signal to $u_{0}\in\mathcal{U}_{k}$ is changed from $n$ to $n+1$. For single-tier cellular networks, the expression (which is complicated) for the SIR coverage probability gain of increasing one more DoF for boosting the desired signal to $u_{0}$ has been derived in \cite{li14}, and it has been shown that this gain diminishes as the number of DoF increases. However, the speed that this gain changes has not been characterized in \cite{li14}. In the following lemma, we investigate this gain in HetNets, and characterize the order of this gain when $\tau\to0$.  
\begin{lemma}\label{lem:condiCP_lowbeta}
When $\tau\to0$, we have\footnote{$f(x)=\Theta\left(g(x)\right)$ means that $\lim_{x\to0}\frac{f(x)}{g(x)}=c$ where $0<c<\infty$.} $\mathcal{T}_{k,R_{1k},R_{2k}}\left(n,r_{1k},r_{2k},2^{{\rm E}\left[L_{0,j_{k}}\right]\tau/W}-1\right)=\Theta\left(\tau^{n}\right)$.
\end{lemma}
\begin{proof}
See Appendix \ref{proof:Tn_lowbeta}.
\end{proof}

From \emph{Lemma \ref{lem:condiCP_lowbeta}}, we see that when $\tau\to0$, the gain $\mathcal{T}_{k,R_{1k},R_{2k}}\left(n,r_{1k},r_{2k},2^{{\rm E}\left[L_{0,j_{k}}\right]\tau/W}-1\right)$ decreases as $n$ increases, and the order of $\mathcal{T}_{k,R_{1k},R_{2k}}\left(n,r_{1k},r_{2k},2^{{\rm E}\left[L_{0,j_{k}}\right]\tau/W}-1\right)$ is $\tau^{n}$. Based on \emph{Lemma \ref{lem:condiCP_lowbeta}}, we obtain the order of the rate coverage probability loss of a macro-user, i.e., $\left|\Delta\mathcal{\bar{R}}_{{\rm IN},1}(U,\tau)\right|$, which is shown in the following proposition:
\begin{proposition}\label{prop:CPloss_macro}
When $\tau\to0$, we have $\left|\Delta\mathcal{\bar{R}}_{{\rm IN},1}(U,\tau)\right|=\Theta\left(\tau^{N_1-U}\right)$.
\end{proposition}
\begin{proof}
Follows by showing the integrand in (\ref{eq:deltaS1_general}) is upper bounded by an integrable function. In particular, for the integrand in (\ref{eq:deltaS1_general}), we have 
\begin{align}
\mathcal{T}_{1,y}(N_{1}-U,y,\hat{\beta})f_{Y_{1}}(y)<&\sum_{n_{1}=0}^{N_{1}-U}\sum_{(p_{a})_{a=1}^{n_{1}}\in\mathcal{M}_{n_{1}}}\sum_{(p_{b})_{b=1}^{N_{1}-U-n_{1}}\in\mathcal{M}_{N_{1}-U-n_{1}}}g\left(n_{1},\{p_{a}\},\{p_{b}\},\hat{\beta}\right) \notag\\
&\times \exp\left(-cy^{2}\right)y^{2\sum_{a=1}^{n_{1}}p_{a}+\frac{2\alpha_{1}}{\alpha_{2}}\sum_{b=1}^{N_{1}-U-n_{1}}p_{b}+1}
\end{align}
where $c$ is a real positive constant and $g\left(n_{1},\{p_{a}\},\{p_{b}\},\hat{\beta}\right)$ is the coefficient (independent of $y$). Here, the inequality is obtained by noting that $\mathcal{L}_{I_{j}}\left(s,r_{jk}\right)<1$, $B^{'}(a,b,z)<B(a,b)$ which is the beta function, and $\exp\left(-\pi \lambda_{2}\left(\frac{P_{2}B}{P_{1}}\right)^{\frac{2}{\alpha_{2}}}y^{\frac{2\alpha_{1}}{\alpha_{2}}}y^{\frac{2\alpha_{1}}{\alpha_{2}}}\right)<1$. It can be easily shown that $y^{2\sum_{a=1}^{n_{1}}p_{a}+\frac{2\alpha_{1}}{\alpha_{2}}\sum_{b=1}^{N_{1}-U-n_{1}}p_{b}+1}\exp\left(-cy^{2}\right)$ is integrable. From \emph{Lemma \ref{lem:condiCP_lowbeta}}, we know 
\begin{align}
\mathcal{T}_{k,R_{1k},R_{2k}}\left(N_{1}-U,r_{1k},r_{2k},2^{{\rm E}\left[L_{0,j_{k}}\right]\tau/W}-1\right)=\Theta\left(\tau^{N_{1}-U}\right)\;,
\end{align} 
then using \emph{dominated convergence theorem}, the proof completes.
\end{proof}

\emph{Proposition \ref{prop:CPloss_macro}} shows that when $\tau\to0$, the rate coverage probability loss of a macro-user, i.e., $\left|\Delta\mathcal{\bar{R}}_{{\rm IN},1}(U,\tau)\right|$ in (\ref{eq:CPdelta_diffU}), decreases with $N_{1}-U$, and the decrease is in the order of $\tau^{N_{1}-U}$. Furthermore, for a fixed $N_{1}$, $\left|\Delta\mathcal{\bar{R}}_{{\rm IN},1}(U,\tau)\right|$ increases as $U$ increases. 

\begin{figure}[t]
\centering
\subfigure[$B=2.5$ dB]{
\includegraphics[width=2in]{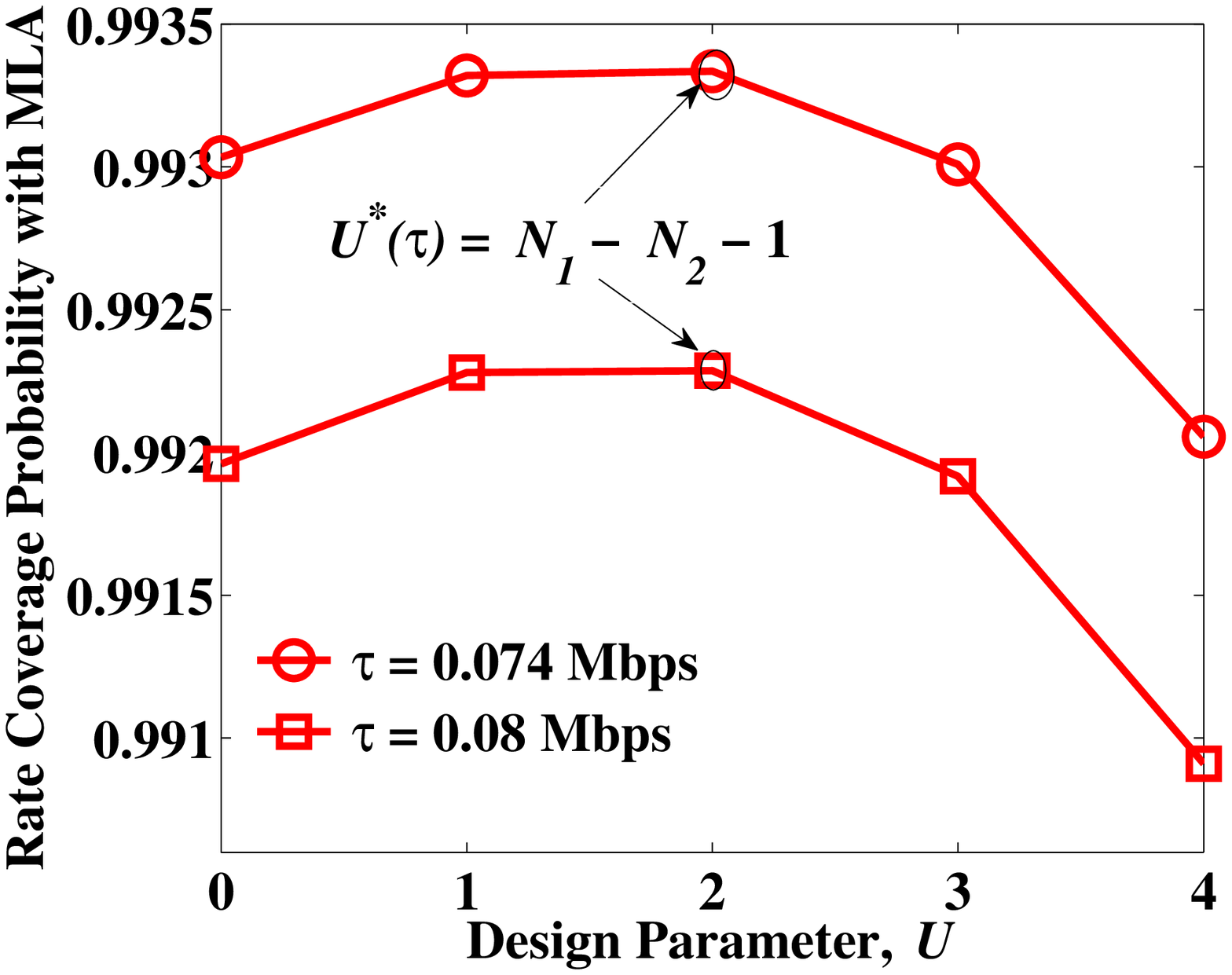}
\label{fig:CPratevsU_smallbeta_B2p5rdB}
}
\subfigure[$B=4.6$ dB]{
\includegraphics[width=2in]{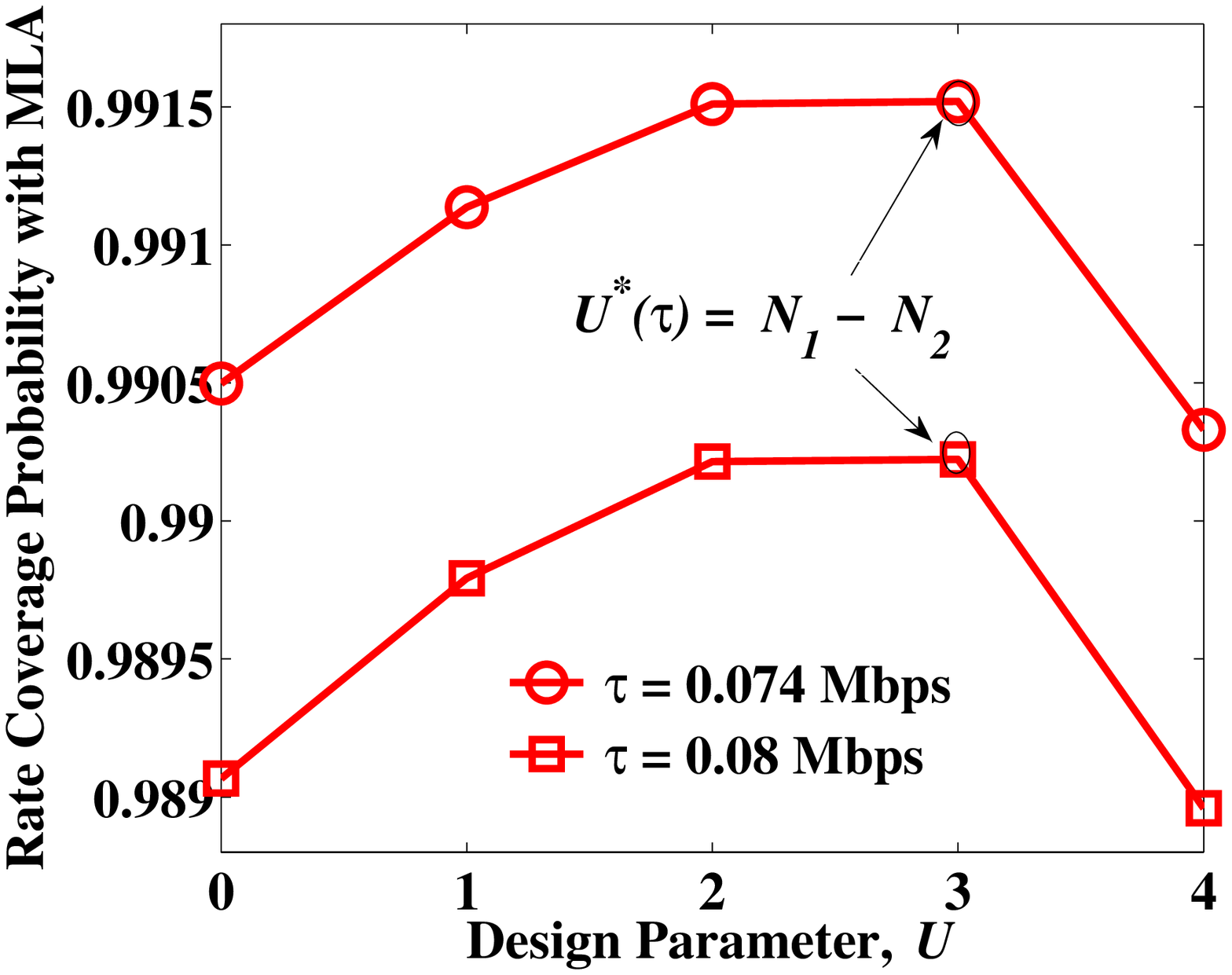}
\label{fig:CPratevsU_smallbeta_B4p6dB}
}
\subfigure[$U^{*}(\tau)$ vs. $\tau$]{
\includegraphics[width=2in]
{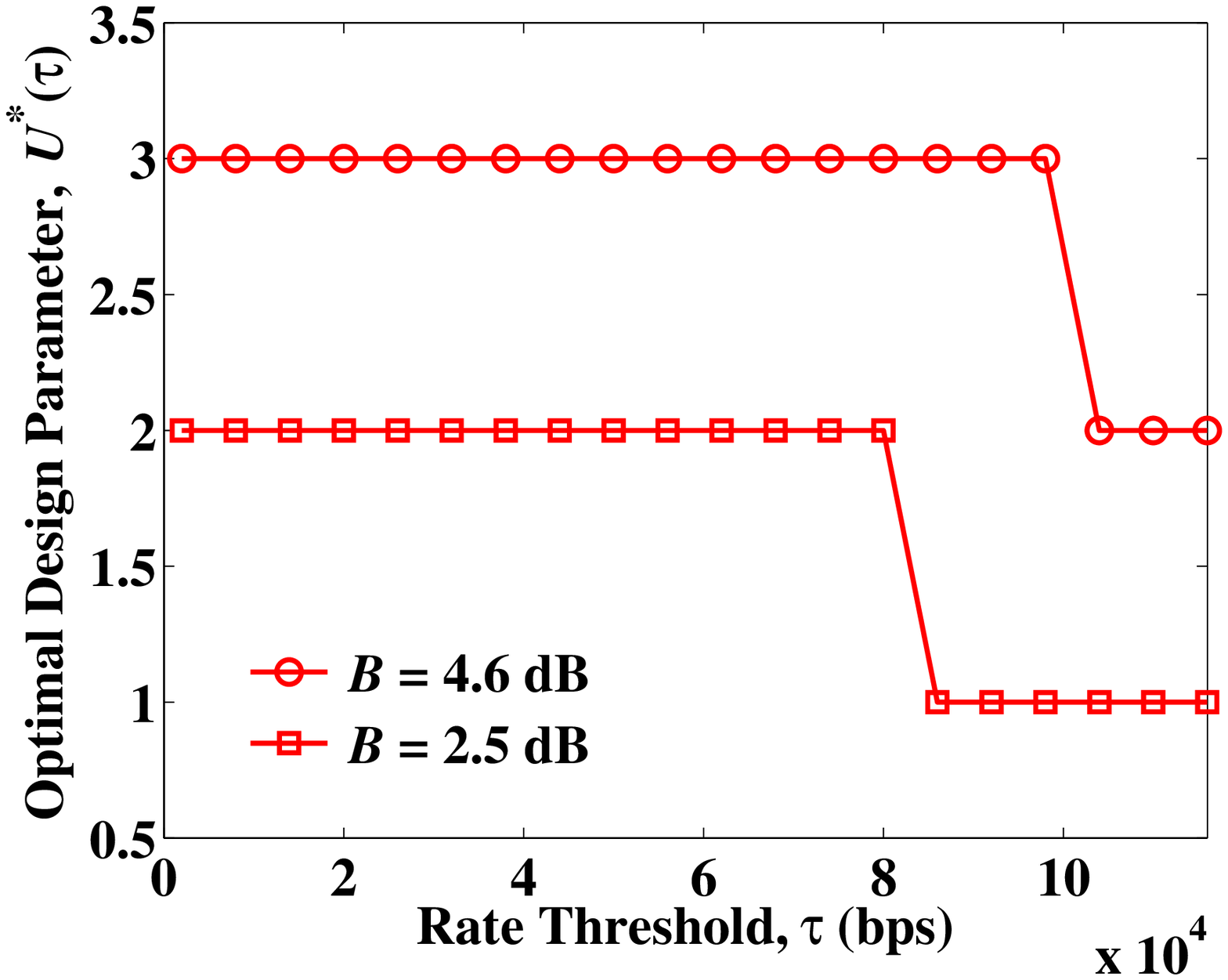}
\label{fig:optUvstau}
}
\caption{\scriptsize Optimal design parameter $U^{*}(\tau)$, at $\alpha_{1}=\alpha_{2}=3$, $\frac{P_{1}}{P_{2}}=10$ dB, $W=10\times 10^{6}$ Hz, $N_{1}=5$, $N_{2}=2$, $\lambda_{u}=0.01$ nodes/m$^{2}$, $\lambda_{1}=0.0001$ nodes/m$^{2}$, and $\lambda_{2}=0.0015$ nodes/m$^{2}$.}\label{fig:CPrateMLAvsU_diffB_smalltau}
\vspace{-8mm}
\end{figure}

Next, we characterize the rate coverage probability gain achieved by an offloaded user, i.e., $\Delta\mathcal{\bar{R}}_{{\rm IN},2O}(U,\tau)$. Using a mean interference-to-signal ratio based approach proposed in \cite{haenggi14}, we obtain the order of $\Delta\mathcal{\bar{R}}_{{\rm IN},2O}(U,\tau)$, which is shown as follows:
\begin{proposition}\label{prop:CPgain_offloadIN}
When $\tau\to0$, we have $\Delta\mathcal{\bar{R}}_{{\rm IN},2O}(U,\tau)=\Theta\left(\tau^{N_{2}}\right)$.
\end{proposition}
\begin{proof}
See Appendix \ref{proof:gain_offload_lowbeta}.
\end{proof}

From \emph{Proposition \ref{prop:CPgain_offloadIN}}, we see that when $\tau\to0$, as the number of antennas at each pico-BS $N_{2}$ increases, the rate coverage probability gain of an offloaded user, i.e., $\Delta\mathcal{\bar{R}}_{{\rm IN},2O}(U,\tau)$ in (\ref{eq:CPdelta_diffU}), decreases, and the decrease is in the order of $\tau^{N_{2}}$.

According to (\ref{eq:CPdelta_diffU}), \emph{Proposition \ref{prop:CPloss_macro}} and \emph{Proposition \ref{prop:CPgain_offloadIN}}, and noting that $\mathcal{A}_{2O}$ and $\mathcal{A}_{1}$ are independent of $\tau$, we have 
\begin{align}\label{eq:CPchange_order}
\Delta\mathcal{\bar{R}}_{\rm IN}(U,\tau)=\Theta\left(\tau^{N_{2}}\right)-\Theta\left(\tau^{N_{1}-U}\right)=
\begin{cases}
\Theta\left(\tau^{N_{2}}\right)>0\;,\hspace{13.5mm}{\rm when}\hspace{2mm}U<N_{1}-N_{2}\\
\Theta\left(\tau^{N_{2}}\right)-\Theta\left(\tau^{N_{2}}\right)\;,\hspace{1mm}{\rm when}\hspace{2mm}U=N_{1}-N_{2}\\
\Theta\left(\tau^{N_{2}-U}\right)<0\;,\hspace{9mm}{\rm when}\hspace{2mm}U>N_{1}-N_{2}\\
\end{cases}
.
\end{align} Since $U^{*}(\tau)$ satisfies $\Delta\mathcal{\bar{R}}_{\rm IN}(U^{*}(\tau),\tau)>0$ and $\Delta\mathcal{\bar{R}}_{\rm IN}(U^{*}(\tau)+1,\tau)\le0$, we see from (\ref{eq:CPchange_order}) that $U^{*}(\tau)$ should be in the set $\{N_{1}-N_{2}-1,N_{1}-N_{2}\}$, and the exact value of $U^{*}(\tau)$ depends on whether $\Delta\mathcal{\bar{R}}_{\rm IN}(U,\tau)$ is positive or not when $U=N_{1}-N_{2}$ (i.e., the second case in (\ref{eq:CPchange_order})), i.e., whether the coefficient in $\Theta\left(\tau^{N_{2}}\right)$ corresponding to $\mathcal{A}_{2O}\Delta\mathcal{\bar{R}}_{{\rm IN},2O}(U,\tau)$ (i.e., the first one) is larger than that in $\Theta\left(\tau^{N_{2}}\right)$  corresponding to $\mathcal{A}_{1}\left|\Delta\mathcal{\bar{R}}_{{\rm IN},1}(U,\tau)\right|$ (i.e., the second one) or not. 
According to the above discussions, we can obtain the following theorem: 
\begin{theorem}\label{prop:optU_lowbeta}
When $\tau\to0$, the optimal design parameter $U^{*}(\tau)\to U_{0}^{*}$, where $U_{0}^{*}\in\{N_{1}-N_{2}-1,N_{1}-N_{2}\}$.
\end{theorem}

\emph{Theorem \ref{prop:optU_lowbeta}} shows that when $\tau\to0$, the optimal design parameter $U^{*}(\tau)$ converges to a fixed value in the set $\{N_{1}-N_{2}-1,N_{1}-N_{2}\}$, which is only related to the number of antennas at each macro-BS and each pico-BS. This is because when $\tau\to0$, the ``gain" $\mathcal{A}_{2O}\Delta\mathcal{\bar{R}}_{{\rm IN},2O}(U,\tau)=\Theta\left(\tau^{N_{2}}\right)$ and the ``penalty" $\mathcal{A}_{1}\left|\Delta\mathcal{\bar{R}}_{{\rm IN},1}(U,\tau)\right|=\Theta\left(\tau^{N_{1}-U}\right)$.  

Figs.\ \ref{fig:CPratevsU_smallbeta_B2p5rdB} and \ref{fig:CPratevsU_smallbeta_B4p6dB} plot $\mathcal{\bar{R}}_{\rm IN}(U,\tau)$ vs. the design parameter $U$ for different bias factors $B$. We see that when $B=2.5$ dB, $U^{*}(\tau)=N_1-N_{2}-1=2$; when $B=4.6$ dB, $U^{*}(\tau)=N_1-N_{2}=3$ (note that $U^{*}(\tau)$ increases with $B$). Moreover, Fig.\ \ref{fig:optUvstau} plots the optimal design parameter $U^{*}(\tau)$ vs. rate threshold $\tau$ for different bias factors $B$, from which we see that $U^{*}(\tau)$ converges to a fixed value $U_{0}^{*}\in\{N_{1}-N_{2}-1,N_{1}-N_{2}\}$ when $\tau$ is sufficiently small (e.g., $\tau< 0.1$ Mbps for $B=4.6$ dB). These observations verify \emph{Theorem \ref{prop:optU_lowbeta}}.

\subsection{Rate Coverage Probability Optimization for General $\tau$} 
In this part, we discuss the optimality property of $U^{*}(\tau)$ for general $\tau$. Note that, for general $\tau$, the ``gain" $\mathcal{A}_{2O}\Delta\mathcal{\bar{R}}_{{\rm IN},2O}(U,\tau)\neq\Theta\left(\tau^{N_{2}}\right)$ and the ``penalty" $\mathcal{A}_{1}\left|\Delta\mathcal{\bar{R}}_{{\rm IN},1}(U,\tau)\right|\neq\Theta\left(\tau^{N_{1}-U}\right)$. Hence, different from the case for small $\tau$, for general $\tau$, $U^{*}(\tau)$ also depends on other system parameters besides $N_1$ and $N_2$. Fig.\ \ref{fig:CPrateMLAvsU_diffB} plots the rate coverage probability with MLA vs. $U$ for different bias factors $B$. We can see that besides $N_{1}-N_{2}-1$ and $N_{1}-N_{2}$, $U^{*}(\tau)$ can also take other values in set $\{0,1,\ldots,N_{1}-1\}$. In particular, we see that $U^{*}(\tau)$ can be $0$ (at $B=2$ dB), $2$ (at $B=5$ dB), and $4$ (at $B=10$ dB). Interestingly, similar to the case for small $\tau$ in Fig.\ \ref{fig:CPrateMLAvsU_diffB_smalltau}, from Fig.\ \ref{fig:CPrateMLAvsU_diffB}, for general $\tau$, we can also see that $U^{*}(\tau)$ increases with the bias factor $B$. 

\begin{figure}[t]
\centering
\subfigure[$B=2$ dB]{
\includegraphics[width=2in]
{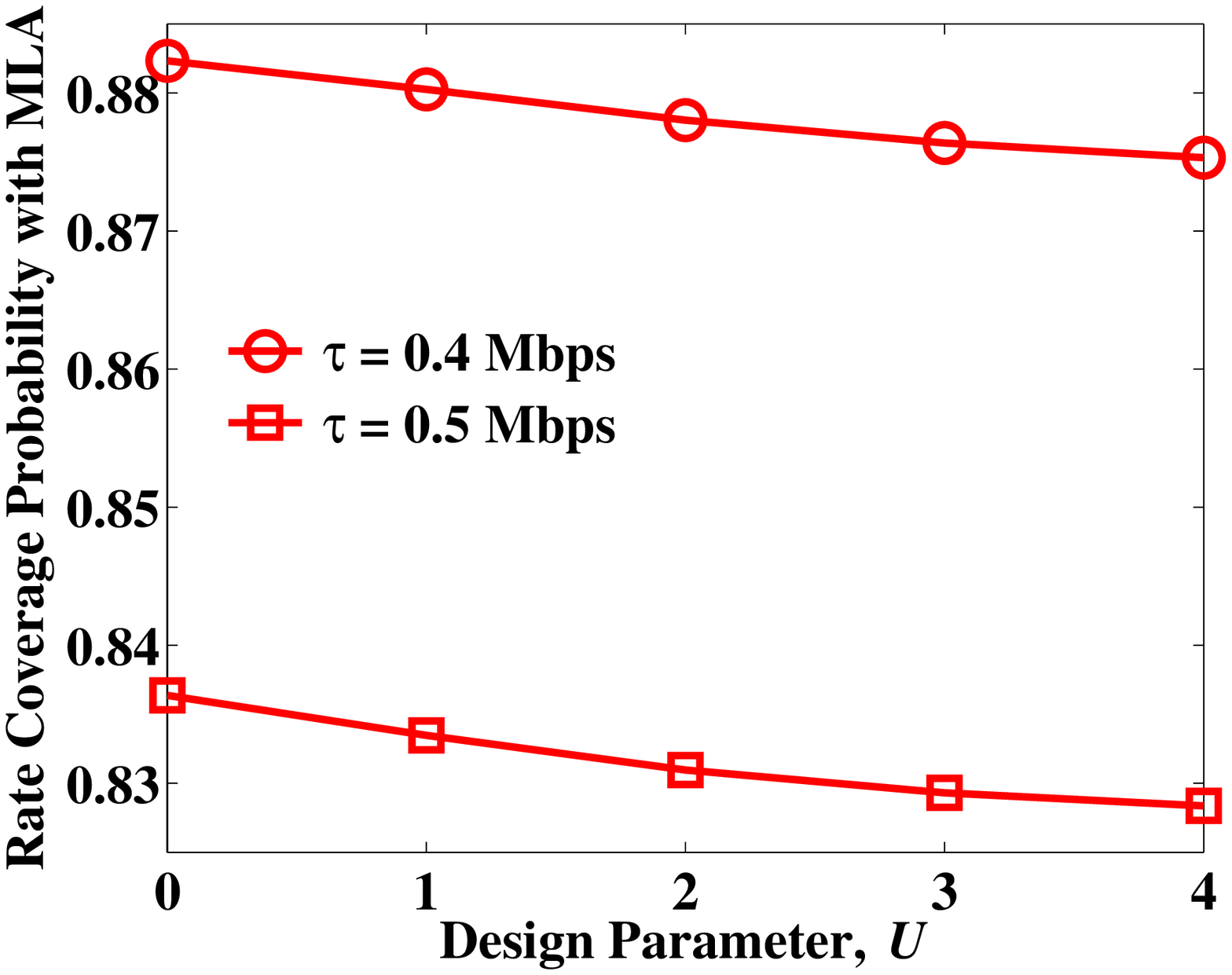}
\label{fig:CPratevsU_B2dB}
}
\subfigure[$B=5$ dB]{
\includegraphics[width=2in]{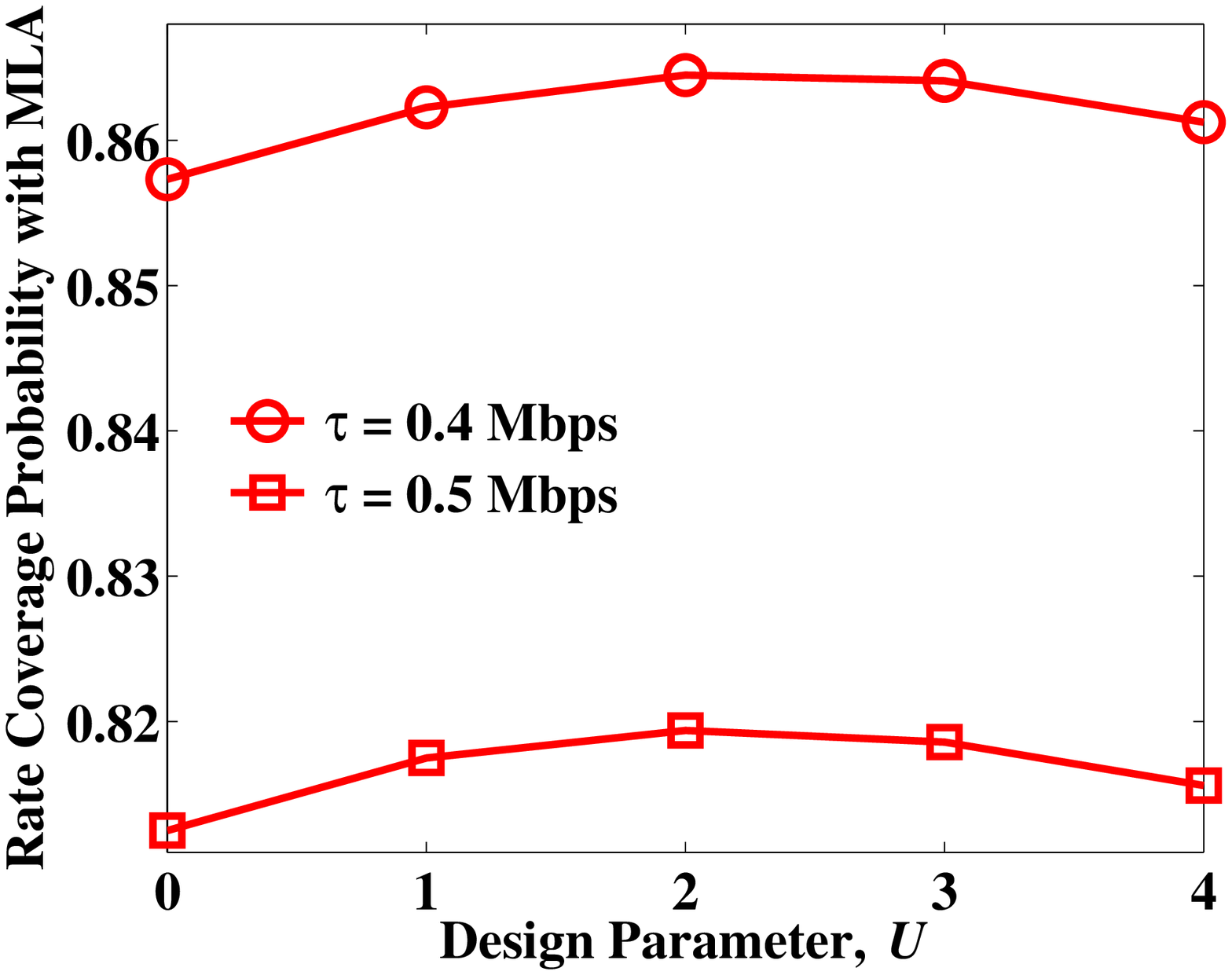}
\label{fig:CPratevsU_B5dB}
}
\subfigure[$B=10$ dB]{
\includegraphics[width=2in]{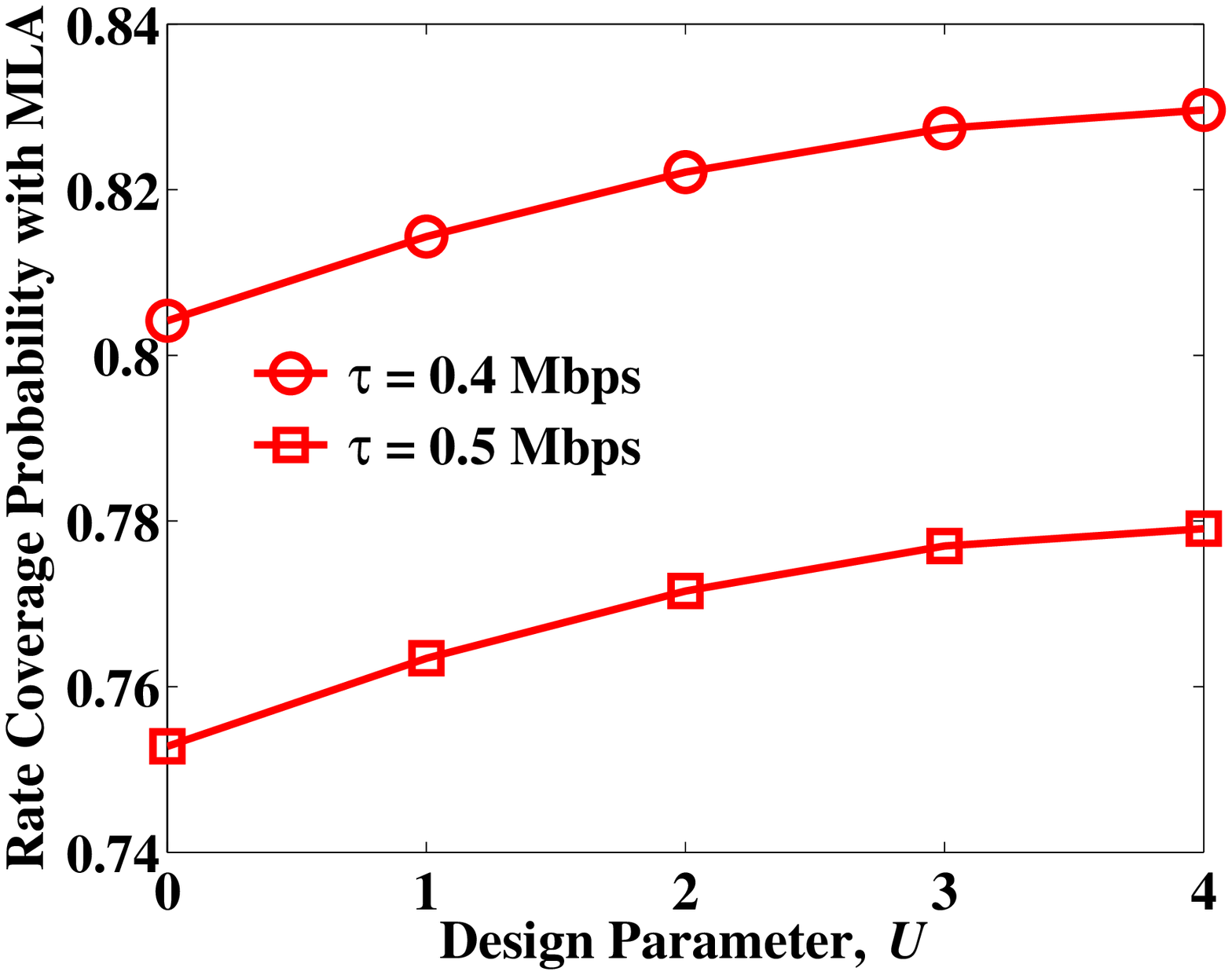}
\label{fig:CPratevsU_B10dB}
}
\caption{\scriptsize Rate coverage probability with MLA vs. $U$ for different bias factors $B$ and general $\tau$, at $\alpha_{1}=\alpha_{2}=3$, $\frac{P_{1}}{P_{2}}=10$ dB, $N_{1}=5$, $N_{2}=2$, $W=10\times 10^{6}$ Hz, $\lambda_{1}=0.0001$ nodes/m$^{2}$, $\lambda_{2}=0.0015$ nodes/m$^{2}$, and $\lambda_{u}=0.01$ nodes/m$^{2}$.}\label{fig:CPrateMLAvsU_diffB}
\vspace{-8mm}
\end{figure}

\section{Rate Coverage Probability Comparison}

In this section, we first analyze the rate coverage probabilities of the simple offloading scheme without interference management (i.e., $U=0$) and the multi-antenna version of ABS in 3GPP-LTE \cite{singh13}. Then, we compare the rate coverage probability of each user type and the overall rate coverage probability of the IN scheme with those of the simple offloading scheme and ABS. 


\subsection{Rate Coverage Probability Analysis for Simple Offloading Scheme and ABS}
\subsubsection{Analysis for Simple Offloading Scheme (i.e., $U=0$)}
Note that $U=0$ is a special case of the IN scheme (under a given $U\in\{0,1,\ldots,N_{1}-1\}$). As such, by letting $U=0$ in \emph{Theorem \ref{cor:CPrate_overall}} and \emph{Corollary \ref{cor:CPrate_MLA}}, we can obtain the rate coverage probability and its MLA of the simple offloading scheme, respectively. In addition, from the resultant expressions, we can know the rate coverage probabilities of the macro-users $\mathcal{\bar{R}}_{U=0,1}\left(\tau\right)$, the unoffloaded pico-users $\mathcal{\bar{R}}_{U=0,2\bar{O}}(\tau)$ and the offloaded users $\mathcal{\bar{R}}_{U=0,2O}\left(\tau\right)$, where $\mathcal{\bar{R}}_{U=0,k}(\tau)=\mathcal{\bar{R}}_{{\rm IN},k}(0,\tau)$ $\left(k\in\{1,2O\}\right)$ and $\mathcal{\bar{R}}_{U=0,2\bar{O}}(\tau)=\mathcal{R}_{{\rm IN},2\bar{O}}(\tau)$. Here, we omit the expressions of the rate coverage probability and its MLA of the simple offloading scheme. Note that \cite{gupta13} also derived the rate coverage probability and its MLA of the macro-users and the pico-users under the simple offloading scheme in large multi-antenna HetNets. However, they did not further obtain the results for the unoffloaded pico-users and the offloaded (pico-) users.  



\subsubsection{Analysis for ABS}
We consider ABS with a given design parameter $\eta\in(0,1)$. Specifically, in ABS, $\eta$ fraction of the resource $W$ is utilized by the pico-BSs to serve offloaded users only, while the remaining $1-\eta$ fraction of the resource $W$ is utilized simultaneously by the macro-BSs and pico-BSs to serve the macro-users and unoffloaded pico-users, respectively \cite{singh13}. In other words, to avoid interference to the offloaded users from all the macro-BSs, the resource used at each BS to serve its associated users in ABS is reduced due to the resource partition (parameterized by $\eta$). Note that different from ABS, in the IN scheme and the simple offloading scheme, each BS utilizes all the resource $W$ to serve its associated users. Similar to (\ref{eq:CPrate_each_IN}), the rate coverage probability of $u_{0}\in\mathcal{U}_{k}$ $\left(k\in\{1,2\bar{O},2O\}\right)$ under ABS is defined as 
\begin{align}\label{eq:rate_ABS_def}
\mathcal{R}_{{\rm ABS},k}(\eta,\tau)&\define{\rm Pr}\left(R_{{\rm ABS},k,0}>\tau|u_{0}\in\mathcal{U}_{k}\right)\notag\\
&={\rm Pr}\left(\frac{\eta_{k}W}{L_{0,k}}\log_{2}\left(1+{\rm SIR}_{{\rm ABS},k,0}\right)>\tau|u_{0}\in\mathcal{U}_{k}\right)
\end{align} 
where $R_{{\rm ABS},k,0}$ and ${\rm SIR}_{{\rm ABS},k,0}$ denote the rate and SIR of $u_{0}\in\mathcal{U}_{k}$ in ABS, respectively, $\eta_{1}=\eta_{2\bar{O}}=1-\eta$, and $\eta_{2O}=\eta$. Similar to (\ref{eq:CPrate_CP}), the rate coverage probability of ABS under the design parameter $\eta$ can be written as: 
\begin{align}
\mathcal{R}_{{\rm ABS}}(\eta,\tau)&\define{\rm Pr}\left(R_{{\rm ABS},0}>\tau\right)\notag\\
&=\sum_{k\in\{1,2\bar{O},2O\}}\mathcal{A}_{k}\mathcal{R}_{{\rm ABS},k}(\eta,\tau)\;, 
\end{align}
where $R_{{\rm ABS},0}$ is the rate of $u_{0}$ (which can be in any user set) in ABS. Applying similar methods in calculating the rate coverage probability and its MLA of the IN scheme in \emph{Theorem \ref{cor:CPrate_overall}} and \emph{Corollary \ref{cor:CPrate_MLA}}, we can obtain the rate coverage probability and its MLA of ABS, respectively. In particular, the rate coverage probability of ABS $\mathcal{R}_{\rm ABS}(\eta,\tau)$  under $\eta\in(0,1)$ is given as follows:
\begin{proposition}\label{prop:CPrate_ABS}
The rate coverage probability of ABS under $\eta\in(0,1)$ is
\begin{align}\label{eq:CPrate_ABS}
\mathcal{R}_{\rm ABS}(\eta,\tau)=\mathcal{A}_{1}\mathcal{R}_{{\rm ABS},1}(\eta,\tau)+\mathcal{A}_{2\bar{O}}\mathcal{R}_{{\rm ABS},2\bar{O}}(\eta,\tau)+\mathcal{A}_{2O}\mathcal{R}_{{\rm ABS},2O}(\eta,\tau)
\end{align} 
where $\mathcal{A}_{1}$, $\mathcal{A}_{2\bar{O}}$ and $\mathcal{A}_{2O}$ are given in \emph{Theorem \ref{cor:CP_uncondi}}, and 
\begin{align}
&\mathcal{R}_{{\rm ABS},1}(\eta,\tau)=\sum_{u\ge1}{\rm Pr}\left(L_{0,1}=u\right)\int_{0}^{\infty}\sum_{n=0}^{N_{1}-1}\frac{1}{n!}\sum_{n_{1}=0}^{n}\binom{n}{n_{1}}\tilde{\mathcal{L}}_{I_{1}}^{(n_{1})}\left(s,y\right)\Big|_{s=\hat{\beta}_{1}(\eta,u) y^{\alpha_{1}}}\notag\\
&\hspace{42mm}\times\tilde{\mathcal{L}}_{I_{2}}^{(n-n_{1})}\left(s,\left(P_{2}B/P_{1}\right)^{\frac{1}{\alpha_{2}}}y^{\frac{\alpha_{1}}{\alpha_{2}}}\right)\Big|_{s=\hat{\beta}_{1}(\eta,u) y^{\alpha_{1}}\frac{P_{2}}{P_{1}}}f_{Y_{1}}(y){\rm d}y\;,\\
&\mathcal{R}_{{\rm ABS},2\bar{O}}(\eta,\tau)=\sum_{u\ge1}{\rm Pr}\left(L_{0,2\bar{O}}=u\right)\int_{0}^{\infty}\mathcal{S}_{{\rm IN},2\bar{O},Y_{2}}\left(y,\hat{\beta}_{2\bar{O}}(\eta,u)\right)f_{Y_{2}}(y){\rm d}y\;,\\
&\mathcal{R}_{{\rm ABS},2O}(\eta,\tau)=\sum_{u\ge1}{\rm Pr}\left(L_{0,2O}=u\right)\int_{0}^{\infty}\sum_{n=0}^{N_{2}-1}\frac{1}{n!}\tilde{\mathcal{L}}^{(n)}_{I_{2}}\left(s,y\right)\Big|_{s=\hat{\beta}_{2O}(\eta,u)y^{\alpha_{2}}}f_{Y_{2O}}(y){\rm d}y\;.
\end{align} 
Here, $\hat{\beta}_{1}(\eta,u)=2^{\frac{\tau u}{W(1-\eta)}}-1$, $\hat{\beta}_{2\bar{O}}(\eta,u)=2^{\frac{\tau u}{W(1-\eta)}}-1$, $\hat{\beta}_{2O}(\eta,u)=2^{\frac{\tau u}{W\eta}}-1$, ${\rm Pr}\left(L_{0,1}=u\right)$ is given in \emph{Theorem \ref{cor:CPrate_overall}}, $f_{Y_{1}}(y)$ and $f_{Y_{2}}(y)$ are given in \emph{Theorem \ref{cor:CP_uncondi}}, and 
\begin{align}
{\rm Pr}\left(L_{0,2\bar{O}}=u\right)=&\frac{3.5^{3.5}\Gamma\left(u+3.5\right)\left(\frac{\lambda_{u}\mathcal{A}_{2\bar{O}}}{\lambda_{2}}\right)^{u-1}}{\Gamma(3.5)(u-1)!}\left(3.5+\frac{\lambda_{u}\mathcal{A}_{2\bar{O}}}{\lambda_{2}}\right)^{-(u+3.5)}\;,\quad u\ge1\\
{\rm Pr}\left(L_{0,2O}=u\right)=&\frac{3.5^{3.5}\Gamma\left(u+3.5\right)\left(\frac{\lambda_{u}\mathcal{A}_{2O}}{\lambda_{2}}\right)^{u-1}}{\Gamma(3.5)(u-1)!}\left(3.5+\frac{\lambda_{u}\mathcal{A}_{2O}}{\lambda_{2}}\right)^{-(u+3.5)}\;,\quad u\ge1\\
f_{Y_{2O}}(y)=&\frac{2\pi\lambda_{2}}{\mathcal{A}_{2O}}\left(\exp\left(-\pi\lambda_{1}\left(P_{1}/(P_{2}B)\right)^{\frac{2}{\alpha_{1}}}y^{\frac{2\alpha_{2}}{\alpha_{1}}}\right)-\exp\left(-\pi\lambda_{1}\left(P_{1}/P_{2}\right)^{\frac{2}{\alpha_{1}}}y^{\frac{2\alpha_{2}}{\alpha_{1}}}\right)\right)\notag\\
&\hspace{6mm}\times y\exp\left(-\pi\lambda_{2}y^{2}\right)
\end{align} 
are given by \cite{singh13}.
\end{proposition}
\begin{proof}
Similar to the proof of \emph{Theorem \ref{cor:CPrate_overall}}.
\end{proof}

As shown in Section \ref{subsec:CPrate_anal}, the rate coverage probability with MLA, which has a simpler expression, is sufficiently accurate. The rate coverage probability with MLA for ABS $\mathcal{\bar{R}}_{\rm ABS}(\eta,\tau)$ under $\eta\in(0,1)$ is given as follows:
\begin{proposition}\label{prop:CPrateMLA_ABS}
The rate coverage probability with MLA of ABS under $\eta\in(0,1)$ is
\begin{align}\label{eq:CPrateMLA_ABS}
\mathcal{\bar{R}}_{\rm ABS}(\eta,\tau)=\mathcal{A}_{1}\mathcal{\bar{R}}_{{\rm ABS},1}(\eta,\tau)+\mathcal{A}_{2\bar{O}}\mathcal{\bar{R}}_{{\rm ABS},2\bar{O}}(\eta,\tau)+\mathcal{A}_{2O}\mathcal{\bar{R}}_{{\rm ABS},2O}(\eta,\tau)
\end{align} 
where $\mathcal{A}_{1}$, $\mathcal{A}_{2\bar{O}}$ and $\mathcal{A}_{2O}$ are given in \emph{Theorem \ref{cor:CP_uncondi}}, and 
\begin{align}
&\mathcal{\bar{R}}_{{\rm ABS},1}(\eta,\tau)=\int_{0}^{\infty}\sum_{n=0}^{N_{1}-1}\frac{1}{n!}\sum_{n_{1}=0}^{n}\binom{n}{n_{1}}\tilde{\mathcal{L}}_{I_{1}}^{(n_{1})}\left(s,y\right)\Big|_{s=\tilde{\beta}_{1}(\eta) y^{\alpha_{1}}}\notag\\
&\hspace{36mm}\times\tilde{\mathcal{L}}_{I_{2}}^{(n-n_{1})}\left(s,\left(P_{2}B/P_{1}\right)^{\frac{1}{\alpha_{2}}}y^{\frac{\alpha_{1}}{\alpha_{2}}}\right)\Big|_{s=\tilde{\beta}_{1}(\eta) y^{\alpha_{1}}\frac{P_{2}}{P_{1}}}f_{Y_{1}}(y){\rm d}y\;,\\
&\mathcal{\bar{R}}_{{\rm ABS},2\bar{O}}(\eta,\tau)=\int_{0}^{\infty}\mathcal{S}_{{\rm IN},2\bar{O},Y_{2}}\left(y,\tilde{\beta}_{2\bar{O}}(\eta)\right)f_{Y_{2}}(y){\rm d}y\;,\\
&\mathcal{\bar{R}}_{{\rm ABS},2O}(\eta,\tau)=\int_{0}^{\infty}\sum_{n=0}^{N_{2}-1}\frac{1}{n!}\tilde{\mathcal{L}}^{(n)}_{I_{2}}\left(s,y\right)\Big|_{s=\tilde{\beta}_{2O}(\eta)y^{\alpha_{2}}}f_{Y_{2O}}(y){\rm d}y\;.
\end{align} 
Here, $\tilde{\beta}_{1}(\eta)=2^{\frac{\tau{\rm E}\left[L_{0,1}\right]}{W(1-\eta)}}-1$, $\tilde{\beta}_{2\bar{O}}(\eta)=2^{\frac{\tau{\rm E}\left[L_{0,2\bar{O}}\right]}{W(1-\eta)}}-1$, $\tilde{\beta}_{2O}(\eta)=2^{\frac{\tau{\rm E}\left[L_{0,2O}\right]}{W\eta}}-1$, ${\rm E}\left[L_{0,1}\right]$ is given in \emph{Corollary \ref{cor:CPrate_MLA}}, ${\rm E}\left[L_{0,2\bar{O}}\right]=1+1.28\frac{\lambda_{u}\mathcal{A}_{2\bar{O}}}{\lambda_{2}}$,  ${\rm E}\left[L_{0,2O}\right]=1+1.28\frac{\lambda_{u}\mathcal{A}_{2O}}{\lambda_{2}}$, $f_{Y_{1}}(y)$ and $f_{Y_{2}}(y)$ are given in \emph{Theorem \ref{cor:CP_uncondi}}, and $f_{Y_{2O}}(y)$ is given in \emph{Proposition \ref{prop:CPrate_ABS}}.
\end{proposition}
\begin{proof}
Similar to the proof of \emph{Corollary \ref{cor:CPrate_MLA}}. 
\end{proof}

Note that the rate coverage probability and its MLA of multi-antenna ABS shown in \emph{Proposition \ref{prop:CPrate_ABS}} and \emph{Proposition \ref{prop:CPrateMLA_ABS}} are derived using higher order derivatives of the Laplace transform of the aggregate interference, and can be treated as extensions of the single-antenna results derived using the Laplace transform of the aggregate interference in \cite{singh13}.

\subsection{Rate Coverage Probability Comparison for Each User Type}\label{subset:compABS_fixB}
In this part, we compare the rate coverage probability of $u_{0}\in\mathcal{U}_{k}$ ($k\in\left(1,2\bar{O},2O\right)$) in the IN scheme (under a given $U\in\{0,1,\ldots, N_{1}-1\}$) with those in the simple offloading scheme (i.e., $U=0$) and ABS (under a given $\eta\in(0,1)$), respectively, for a fixed bias factor $B$. 

\subsubsection{Comparison with simple offloading scheme} 
First, we compare the rate coverage probability of $u_{0}\in\mathcal{U}_{k}$ $\left(k\in\{1,2\bar{O},2O\}\right)$ in the IN scheme with that in the simple offloading scheme. We can easily show the following lemma:
\begin{lemma}\label{lem:comp_INU0_eachuser}
For all $U\in\{0,1,\ldots,N_{1}-1\}$, we have: i) $\mathcal{\bar{R}}_{{\rm IN}, 1}(U,\tau)\le \mathcal{\bar{R}}_{U=0, 1}(\tau)$, ii) $\mathcal{\bar{R}}_{{\rm IN}, 2\bar{O}}(\tau)$\\$= \mathcal{\bar{R}}_{U=0, 2\bar{O}}(\tau)$, iii) $\mathcal{\bar{R}}_{{\rm IN}, 2O}(U,\tau)\ge \mathcal{\bar{R}}_{U=0, 2O}(\tau)$. The equalities in i) and ii) hold i.f.f. $U=0$. 
\end{lemma}

Now we compare the IN scheme under $U>0$ with the simple offloading scheme (i.e., $U=0$). \emph{Lemma \ref{lem:comp_INU0_eachuser}} can be interpreted below: i) the IN scheme achieves a smaller rate coverage probability for $u_{0}\in\mathcal{U}_{1}$, since the DoF used to serve $u_{0}$ are reduced by $\min\left(U,u_{2OC,0}\right)$; ii) the IN scheme achieves the same rate coverage probability of $u_{0}\in\mathcal{U}_{2\bar{O}}$ as the simple offloading scheme, since $\mathcal{\bar{R}}_{{\rm IN},2\bar{O}}(\tau)$ is independent of $U$; iii) the IN scheme achieves a larger rate coverage probability for $u_{0}\in\mathcal{U}_{2O}$, since $\min\left(U,u_{2OC,0}\right)$ DoF at the nearest macro-BS of $u_{0}$ are used to avoid dominant macro-interference to its $\min\left(U,u_{2OC,0}\right)$ IN offloaded users.

\subsubsection{Comparison with ABS}
Now, we compare the rate coverage probability of $u_{0}\in\mathcal{U}_{k}$ $(k\in\{1,2\bar{O},2O\})$ in the IN scheme with that in ABS, which is summarized in the following:
\begin{lemma}\label{lem:comp_INABS_eachuser}
i) A sufficient condition for $\mathcal{\bar{R}}_{{\rm IN},1}(U,\tau)>\mathcal{\bar{R}}_{{\rm ABS},1}(\eta,\tau)$ when $N_{1},U\to\infty$ with $\frac{U}{N_{1}}\to\kappa\in(0,1)$ and $\tau\to0$ is $\kappa<\eta$;
ii) the necessary and sufficient condition for $\mathcal{\bar{R}}_{{\rm IN},2\bar{O}}(\tau)>\mathcal{\bar{R}}_{{\rm ABS},2\bar{O}}(\eta,\tau)$ is $\frac{1}{{\rm E}\left[L_{0,2}\right]}>\frac{1-\eta}{{\rm E}\left[L_{0,2\bar{O}}\right]}$;
iii) a necessary condition for $\mathcal{\bar{R}}_{{\rm IN},2O}(U,\tau)>\mathcal{\bar{R}}_{{\rm ABS},2O}(\eta,\tau)$ is $\frac{1}{{\rm E}\left[L_{0,2}\right]}>\frac{\eta}{{\rm E}\left[L_{0,2O}\right]}$.
\end{lemma}
\begin{proof}
See Appendix \ref{proof:compABS_CPrate}. 
\end{proof}

Note that the rate coverage probability of $u_{0}\in\mathcal{U}_{k}$ $\left(k\in\{1,2\bar{O},2O\}\right)$ depends on both the SIR of $u_{0}$ and the average resource used to serve $u_{0}$. Thus, \emph{Lemma \ref{lem:comp_INABS_eachuser}} can be understood below: i) the IN scheme (with DoF fraction $1-\kappa$ and resource fraction $1$ for scheduled $u_{0}$) achieves a larger rate coverage probability for $u_{0}\in\mathcal{U}_{1}$ than ABS (with DoF fraction $1$ and resource fraction $1-\eta$ for scheduled $u_{0}$) if $\kappa<\eta$; ii) The IN scheme achieves a larger rate coverage probability for $u_{0}\in\mathcal{U}_{2\bar{O}}$ i.f.f. the average resource (i.e., $\frac{1}{{\rm E}\left[L_{0,2}\right]}$ under MLA) used to serve $u_{0}$ in the IN scheme is larger than that (i.e., $\frac{1-\eta}{{\rm E}\left[L_{0,2\bar{O}}\right]}$ under MLA) in ABS, as the SIRs of $u_{0}\in\mathcal{U}_{2\bar{O}}$ are the same in both schemes; iii) Note that the SIR of $u_{0}\in\mathcal{U}_{2O}$ in the IN scheme is worse than that in ABS, as $u_{0}\in\mathcal{U}_{2O}$ does not experience any macro-interference in ABS, while $u_{0}\in\mathcal{U}_{2O}$ still experiences macro-interference (except the dominant one) in the IN scheme. Hence, it is possible for the IN scheme to achieve a larger rate coverage probability for $u_{0}\in\mathcal{U}_{2O}$ only when the average resource (i.e., $\frac{1}{{\rm E}\left[L_{0,2}\right]}$ under MLA) used to serve $u_{0}\in\mathcal{U}_{2O}$ in the IN scheme is larger than that (i.e., $\frac{\eta}{{\rm E}\left[L_{0,2O}\right]}$ under MLA) in ABS.

Fig.\ \ref{fig:CPrateVSeta_ABS_IN} plots the rate coverage probability with MLA of the IN scheme at $U=7$, and the rate coverage probability with MLA of ABS vs. $\eta$. Note that under the parameters in Fig.\ \ref{fig:CPrateVSeta_ABS_IN}, we have: i) $\kappa=\frac{U}{N_{1}}=0.7$, ii) $1-\frac{{\rm E}\left[L_{0,2\bar{O}}\right]}{{\rm E}\left[L_{0,2}\right]}\approx 0.09$, and iii) $\frac{{\rm E}\left[L_{0,2O}\right]}{{\rm E}\left[L_{0,2}\right]}\approx 0.12$, with ${\rm E}\left[L_{0,2\bar{O}}\right]\approx28.57$, ${\rm E}\left[L_{0,2O}\right]\approx3.86$ and ${\rm E}\left[L_{0,2}\right]\approx31.43$ calculated according to \emph{Proposition \ref{prop:CPrateMLA_ABS}} and  \emph{Corollary \ref{cor:CPrate_MLA}}. From Fig.\ \ref{fig:CPrateVSeta_ABS_IN}, we observe that i) $\eta>0.7$ is sufficient to achieve $\mathcal{\bar{R}}_{{\rm IN},1}(7,\tau)>\mathcal{\bar{R}}_{{\rm ABS},1}(\eta,\tau)$; ii) $\mathcal{\bar{R}}_{{\rm IN},2\bar{O}}(\tau)>\mathcal{\bar{R}}_{{\rm ABS},2\bar{O}}(\eta,\tau)$ i.f.f. $\eta>0.09\approx1-\frac{{\rm E}\left[L_{0,2\bar{O}}\right]}{{\rm E}\left[L_{0,2}\right]}$; iii) $\mathcal{\bar{R}}_{{\rm IN},2O}(7,\tau)>\mathcal{\bar{R}}_{{\rm ABS},2O}(\eta,\tau)$ when $\eta<0.1<0.12\approx\frac{{\rm E}\left[L_{0,2O}\right]}{{\rm E}\left[L_{0,2}\right]}$. These observations verify \emph{Lemma \ref{lem:comp_INABS_eachuser}}.

\begin{figure}[t] \centering
\includegraphics[width=4.2in]{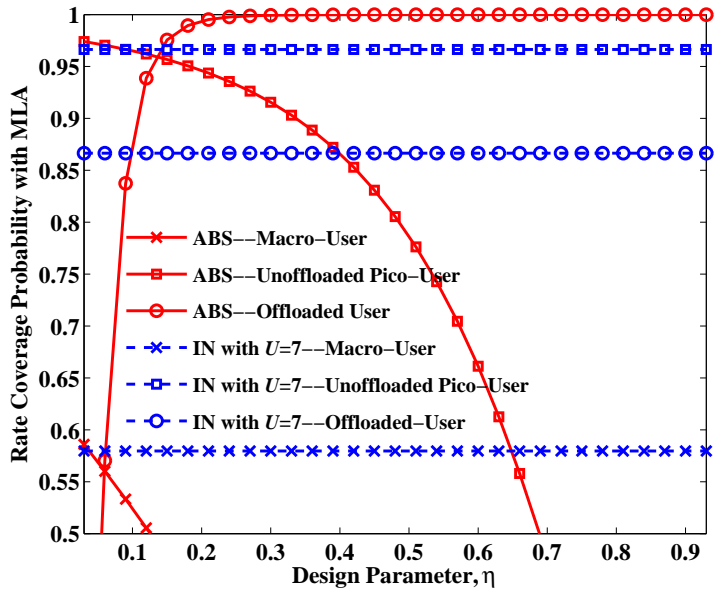}
\caption{\scriptsize Rate coverage probability vs. resource fraction $\eta$ for ABS and IN, at $\frac{P_{1}}{P_{2}}=13$ dB, $W=10$ MHz, $\tau=5\times 10^5$ bps, $N_{1}=10$, $N_{2}=8$, $\lambda_{1}=0.00008$ nodes/m$^{2}$, $\lambda_{2}=0.001$ nodes/m$^{2}$, $\lambda_{u}=0.03$ nodes/m$^{2}$, $\alpha_{1}=4.5$, $\alpha_{2}=4.7$, $B=4$ dB, $\mathcal{A}_{1}\approx0.21$, $\mathcal{A}_{2\bar{O}}\approx 0.72$, and $\mathcal{A}_{2O}\approx0.07$.}\label{fig:CPrateVSeta_ABS_IN}
\vspace{-6mm}
\end{figure}

\begin{figure}[t]
\centering
\subfigure[$N_{1}=8$, $N_{2}=6$, $\eta^{*}(\tau)=0.01$ at $B^{*}_{\rm ABS}$]{
\includegraphics[height=0.45\columnwidth,width=0.45\columnwidth]
{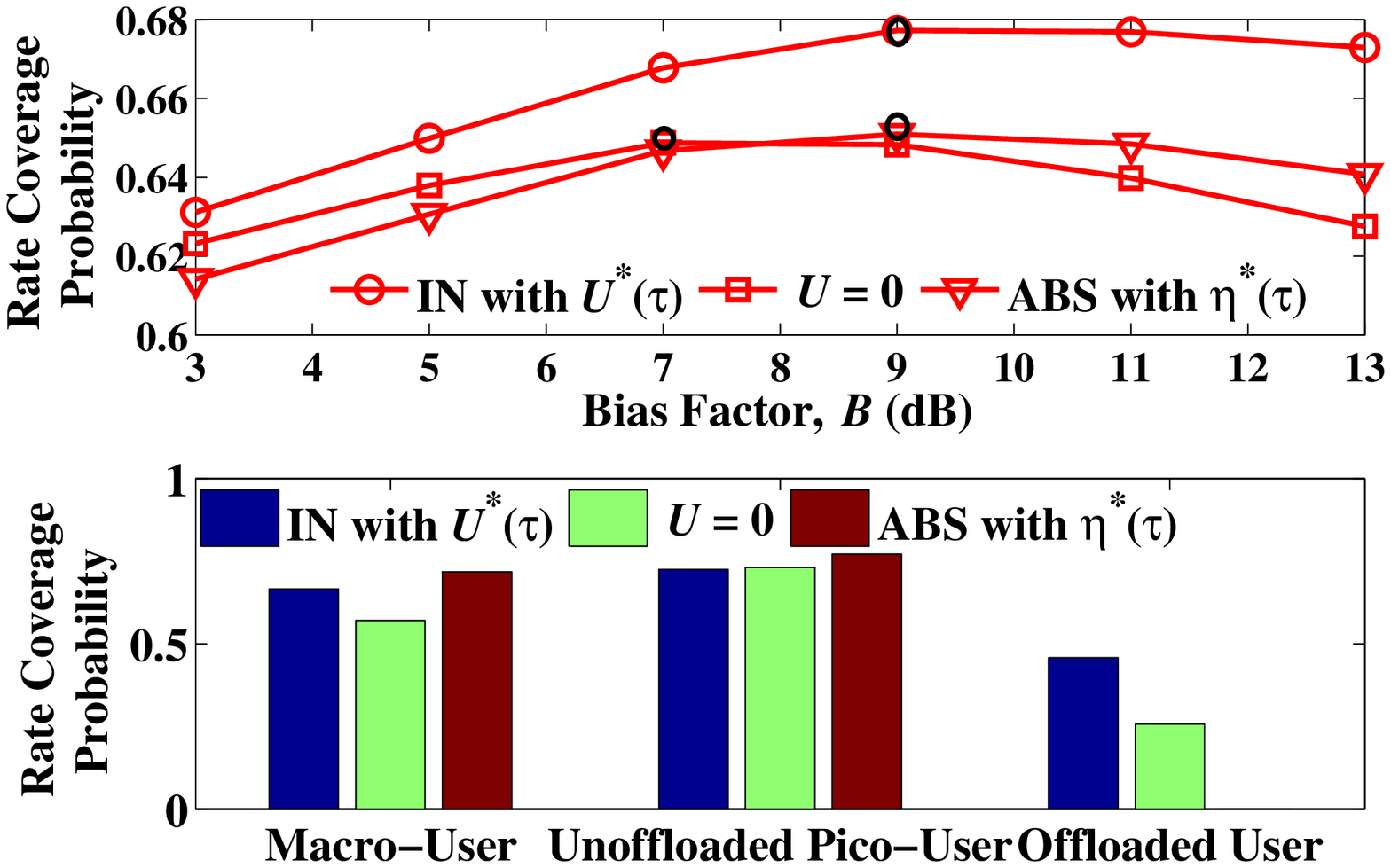}
\label{fig:N18_N26}
}
\subfigure[$N_{1}=18$, $N_{2}=16$, $\eta^{*}(\tau)=0.19$ at $B^{*}_{\rm ABS}$]{
\includegraphics[height=0.45\columnwidth,width=0.45\columnwidth]
{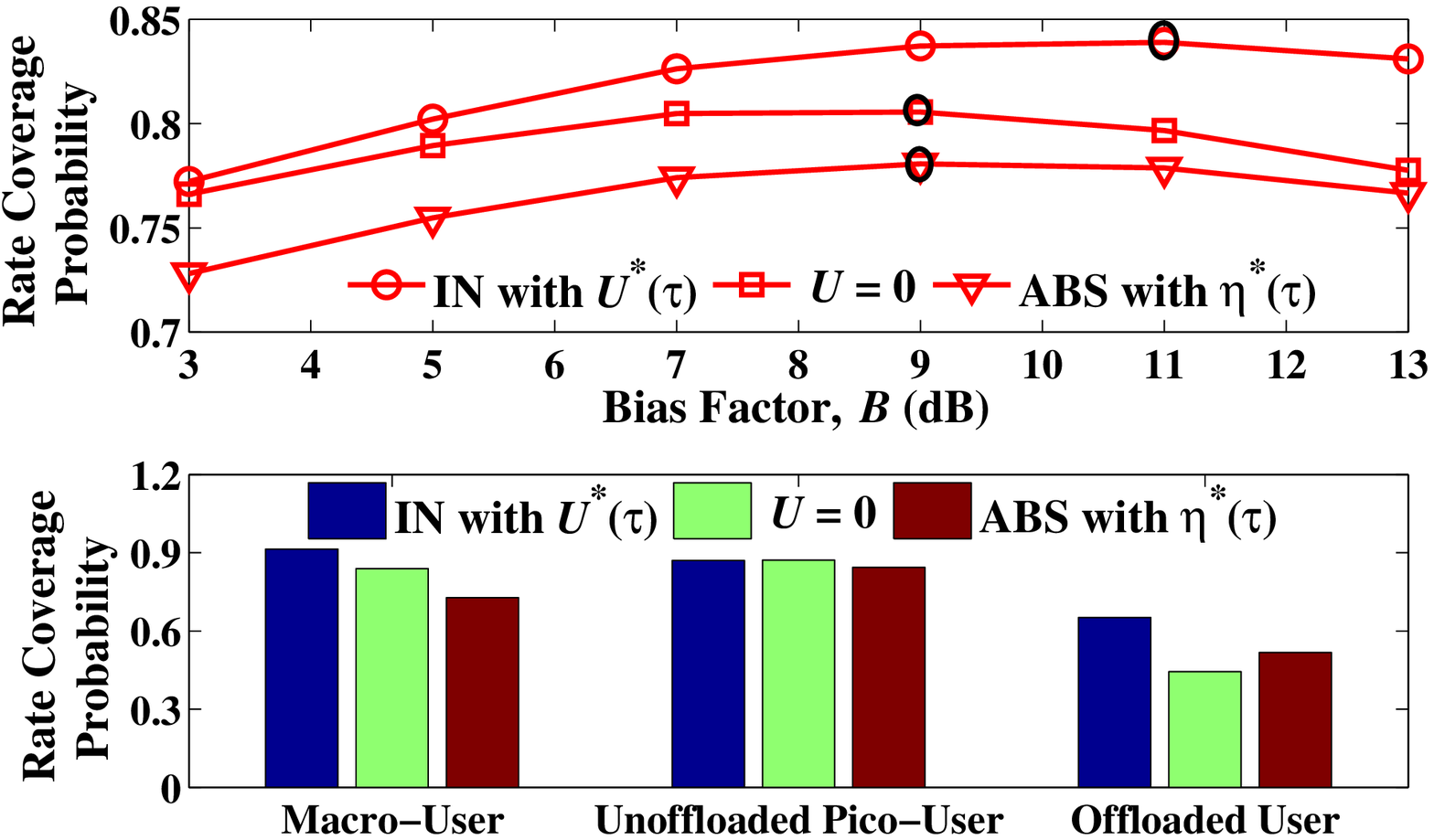}
\label{fig:N118_N216}
}
\caption{\scriptsize Rate coverage probability vs. bias factors $B$, at $\alpha_{1}=4.5$, $\alpha_{2}=4.7$, $\frac{P_{1}}{P_{2}}=13$ dB, $W=10\times 10^{6}$ Hz, $\tau=5\times 10^{5}$ bps, $\lambda_{1}=0.00008$ nodes/m$^{2}$, $\lambda_{2}=0.001$ nodes/m$^{2}$, and $\lambda_{u}=0.05$ nodes/m$^{2}$. In the figures on the top, the points at $B^{*}_{\rm IN}$, $B^{*}_{U=0}$, and $B^{*}_{\rm ABS}$ are highlighted using black ellipse. In the figures at the bottom, the rate coverage probability of each user type in different schemes are plotted at $B^{*}_{\rm IN}$, $B^{*}_{U=0}$, and $B^{*}_{\rm ABS}$, respectively. Note that $\eta^{*}(\tau)$ of ABS is obtained by bisection method with $N_{1}$ iterations, while $U^{*}(\tau)$ of the IN scheme is obtained by exhaustive search over $\{0,1,\ldots,N_{1}-1\}$. }\label{fig:CPratevsBdb_tau0p5MHz}
\vspace{-8mm}
\end{figure}

\subsection{Overall Rate Coverage Probability Comparison}
In this part, we compare the overall rate coverage probability of the \emph{IN scheme under its optimal design parameter $U^{*}(\tau)$} with those of the \emph{simple offloading scheme without interference management (i.e., $U=0$)} and the \emph{multi-antenna version of ABS under its optimal design parameter $\eta^{*}(\tau)\define {\rm arg}\:\max_{\eta\in(0,1)}\mathcal{R}_{\rm ABS}(\eta,\tau)$}. 

First, we compare the rate coverage probability of the IN scheme with that of the simple offloading scheme. Based on the discussions of \emph{Lemma \ref{lem:comp_INU0_eachuser}}, we know that the IN scheme has the benefit of avoiding the dominant macro-interference to the offloaded users. When $B$ is sufficiently large (implying that $\mathcal{A}_{2O}$ is sufficiently large), sufficient offloaded users can benefit from the avoidance of the dominant macro-interference (i.e., the benefit is large). On the other hand, we also know that the loss of the IN scheme compared to the simple offloading scheme is caused by the reduction of the DoF used to serve the macro-users (i.e., at most $\frac{U}{N_{1}}$ reduction of the DoF fraction at each macro-BS in the IN scheme). Thus, when $N_{1}$ is relatively large (implying that the DoF fraction reduction $\frac{U}{N_{1}}$ is minor), the loss due to the DoF reduction is small. Therefore, when $B$ and $N_{1}$ are relatively large (e.g., $B=9$ dB and $N_{1}=8$ in Fig.\ \ref{fig:N18_N26}), the IN scheme can achieve a larger rate coverage probability than the simple offloading scheme.   

Next, we compare the rate coverage probability of the IN scheme with that of ABS. Based on the discussions of \emph{Lemma \ref{lem:comp_INABS_eachuser}}, we know that the benefit of the IN scheme compared to ABS is that it does not have (time or frequency) resource sacrifice. On the other hand, we also know that one loss of the IN scheme compared to ABS is due to the $\frac{U}{N_{1}}$ DoF fraction reduction (as discussed above). Thus, when $N_{1}$ is relatively large (implying that the DoF fraction reduction $\frac{U}{N_{1}}$ is minor), the loss due to the DoF reduction is small. The other loss of the IN scheme compared to ABS is caused by the macro-interference (except the dominant one),  as the IN scheme only avoids the dominant macro-interference to the offloaded users, while ABS avoids all the macro-interference to the offloaded users. When $\alpha_{1}$ is relatively large (implying that the dominant macro-interference is sufficiently strong compared to the remaining macro-interference), the loss due to the remaining macro-interference is small. Therefore, when $N_{1}$ and $\alpha_{1}$ are relatively large (e.g., $N_{1}=8$ and $\alpha_1 = 4.5$ in Fig.\ \ref{fig:N18_N26}), the IN scheme can achieve a larger rate coverage probability than ABS.


The figures on the top of Fig.\ \ref{fig:CPratevsBdb_tau0p5MHz} plot the rate coverage probability vs. the bias factor $B$ for the IN scheme under $U^{*}(\tau)$, the simple offloading scheme, and ABS under $\eta^{*}(\tau)$. We see that the IN scheme achieves a larger rate coverage probability than both the simple offloading scheme and ABS when the bias factor $B$ is relatively large. \footnote{Note that the IN scheme may not provide gains in all scenarios, as suggested in Fig.\ \ref{fig:CPrateMLAvsU_diffB}.} In addition, we consider rate coverage probability maximization over $B$ for these three schemes. We observe that the IN scheme achieves a larger rate coverage probability than both the simple offloading scheme and ABS at their optimal bias factors. Denote the optimal bias factors of the IN scheme, simple offloading scheme and ABS as $B^{*}_{\rm IN}$, $B^{*}_{U=0}$ and $B^{*}_{\rm ABS}$, respectively. We have the following observations for $B^{*}_{\rm IN}$, $B^{*}_{U=0}$ and $B^{*}_{\rm ABS}$. Firstly, $B^{*}_{\rm IN}$, $B^{*}_{U=0}$ and $B^{*}_{\rm ABS}$ are all positive. This implies that the rate coverage probability can be improved by offloading users from the heavily loaded macro-cell tier to the lightly loaded pico-cell tier. Secondly, both $B^{*}_{\rm IN}$ and $B^{*}_{\rm ABS}$ can be larger than $B^{*}_{U=0}$. This implies that the IN scheme and ABS allow more users to be offloaded to the lightly loaded pico-cell tier than the simple offloading scheme, as the IN scheme and ABS can effectively improve the performance of the offloaded users.

We now further investigate the rate coverage probability of the offloaded users, which is one of the main limiting factors for the performance of HetNets with offloading. 
In the IN scheme, the offloaded users do not have (time or frequency) resource sacrifice and dominant macro-interference. However, the offloaded users in ABS suffer from resource limitations, and the offloaded users in the simple offloading scheme suffer from strong interference caused by their dominant macro-interfererence. Hence, the offloaded users in the IN scheme can achieve a larger rate coverage probability than those in both the simple offloading scheme and ABS (e.g., when $\alpha_{1}=4.5$ and $\eta(\tau)=0.01$ in Fig.\ \ref{fig:N18_N26}). The figures at the bottom of Fig.\ \ref{fig:CPratevsBdb_tau0p5MHz} plot the rate coverage probability of three user types at $B^{*}_{\rm IN}$, $B^{*}_{U=0}$, and $B^{*}_{\rm ABS}$, respectively. We can clearly see that the offloaded user in the IN scheme achieves the largest rate coverage probability.

\section{Conclusions}
In this paper, we investigated the IN scheme in downlink two-tier multi-antenna HetNets with offloading. Utilizing tools from stochastic geometry, we first derived a tractable expression for the rate coverage probability of the IN scheme. Then, we considered the rate coverage probability optimization of the IN scheme by solving the optimal design parameter. Finally, we analyzed the performance of the simple offloading scheme without interference management and the multi-antenna version of ABS, and compared the performance of the IN scheme with both of the two schemes in terms of the rate coverage probability of each user type and the overall rate coverage probability. Both the analytical and numerical results showed that the IN scheme can achieve good performance gains over both of the two schemes, especially in the large antenna regime.

\appendix
\subsection{Proof of Lemma \ref{lem:p.m.f._offload_random}}\label{proof:p.m.f._offload_random}
We first note that i) the total number of scheduled pico-users are the same with the total number of pico-BSs, ii) the association area of pico-BSs is $\mathcal{A}_{2}$ fraction of the total area, and iii) the scheduled pico-users are only in the association area of pico-BSs. Hence, the effective density of the scheduled pico-users is $\frac{\lambda_{2}}{\mathcal{A}_{2}}$. 
Next, we approximate the scheduled pico-users as a homogeneous PPP, so that the number of scheduled pico-users in a fixed area is Poisson distributed with density $\frac{\lambda_{2}}{\mathcal{A}_{2}}$. Note that similar approximation approaches are utilized in \cite{bai13,li14IC}. Obviously, the number of active offloaded users in a fixed area is also Poisson distributed with density $\frac{\lambda_{2}}{\mathcal{A}_{2}}$. Further, using the approach in \cite{singh13}, we can calculate the mean of the offloading area (where the offloaded users may reside) of a randomly selected macro-BS, which is $\frac{\mathcal{A}_{2O}}{\lambda_{1}}$. Finally, we obtain (\ref{eq:pmf_UB0}) by following similar steps in calculating the load p.m.f. in \cite{singh13,yu13}. Note that $\mathcal{A}_{2}$ and $\mathcal{A}_{2O}$ are given in \cite{jo12} and \cite{singh13}, respectively.  

\subsection{Proof of Lemma \ref{theo:CP_condi}}\label{proof:theo_CPcondi}
\subsubsection{$k\in\{1,2\bar{O},2OC\}$}
When $k\in\{1,2\bar{O},2OC\}$, based on (\ref{eq:SINR_marco}), (\ref{eq:SINR_pico}), and (\ref{eq:SINR_BC}), we have
{\small\begin{align}\label{eq:CPcondi_twoIterms}
&\mathcal{S}_{{\rm IN},k,R_{1k},R_{2k}}(r_{1k},r_{2k},\beta)={\rm E}_{I_{1},I_{2}}\left[{\rm Pr}\left(\left|\mathbf{h}_{j_{k},00}^{\dagger}\mathbf{f}_{j_{k},0}\right|^{2}>\beta Y_{j_{k}}^{\alpha_{j_{k}}}\left(\frac{P_{1}}{P_{j_{k}}}I_{1}+\frac{P_{2}}{P_{j_{k}}}I_{2}\right)\right)\right]\notag\\
\eqla&\sum_{n=0}^{M_{k}-1}\frac{\left(-\beta Y_{j_k}^{\alpha_{j_k}}\right)^{n}}{n!}\sum_{n_{1}=0}^{n} \binom{n}{n_{1}}\left(\frac{P_{1}}{P_{j_k}}\right)^{n_{1}}\mathcal{L}_{I_{1}}^{(n_{1})}\left(s,r_{1k}\right)\Big|_{s=\beta Y_{j_k}^{\alpha_{j_k}}\frac{P_{1}}{P_{j_k}}} \left(\frac{P_{2}}{P_{j_k}}\right)^{n-n_{1}}\mathcal{L}_{I_{2}}^{(n-n_{1})}\left(s,r_{2k}\right)\Big|_{s=\beta Y_{j_k}^{\alpha_{j_k}}\frac{P_{2}}{P_{j_k}}}
\end{align}}where $(a)$ is obtained by noting that {\small$\left|\mathbf{h}_{j_{k},00}^{\dagger}\mathbf{f}_{j_{k},0}\right|^{2}\dis{\rm Gamma}(M_{k},1)$}, using binomial theorem, and noting that {\small${\rm E}_{I_{j}}\left[I_{j}^{n}\exp\left(-s I_{j}\right)\right]=(-1)^{n}\mathcal{L}_{I_{j}}^{(n)}\left(s,r_{jk}\right)$}.

We now calculate the Laplace transform $\mathcal{L}_{I_{1}}\left(s,r_{1k}\right)$ and its higher order derivative $\mathcal{L}_{I_{1}}^{(m)}\left(s,r_{1k}\right)$. Firstly, let $G_{1,\ell}\define\left|\mathbf{h}_{1,\ell0}^{\dagger}\mathbf{f}_{1,\ell}\right|^{2}$. Then, $\mathcal{L}_{I_{1}}\left(s,r_{1k}\right)$ can be calculated as follows: 
{\small\begin{align}\label{eq:LT}
\mathcal{L}_{I_{1}}\left(s,r_{1k}\right)
=&{\rm E}_{\Phi(\lambda_{1})}\left[\prod_{\ell\in\Phi\left(\lambda_{1}\right)\backslash B\left(0,r_{1k}\right)}{\rm E}_{G_{1,\ell}}\left[\exp\left(-s\frac{1}{\left|D_{1,\ell0}\right|^{\alpha_{j}}}G_{1,\ell}\right)\right]\right]\notag\\
\eqla&\exp\left(-2\pi\lambda_{1}\int_{r_{1k}}^{\infty}\left(1-\frac{1}{1+sr^{-\alpha_{1}}}\right)r{\rm d}r\right)\notag\\
\eqlb&\exp\left(-\frac{2\pi}{\alpha_{1}}\lambda_{1}s^{\frac{2}{\alpha_{1}}}\int_{\frac{1}{1+sr_{1k}^{-\alpha_{1}}}}^{1}\left(1-w\right)^{-\frac{2}{\alpha_{1}}}w^{-1+\frac{2}{\alpha_{1}}}{\rm d}w\right)
\end{align}}where $(a)$ is obtained by utilizing the probability generating functional of PPP \cite{haenggi09}, $(b)$ is obtained by first replacing $s^{-\frac{1}{\alpha_{1}}}r$ with $t$, and then replacing $\frac{1}{1+t^{-\alpha_{1}}}$ with $w$. 

Next, we calculate $\mathcal{L}_{I_{1}}^{(m)}(s,r_{1k})$ based on (\ref{eq:LT}). Utilizing Fa${\rm \grave{a}}$ di Bruno's formula \cite{johnson02}, we have
{\small\begin{align}\label{eq:LT_diff_condi}
\mathcal{L}_{I_{1}}^{(m)}\left(s,r_{1k}\right)
&\hspace{0mm}=\sum_{(p_{a})_{a=1}^{m}\in\mathcal{M}_{m}}\frac{\mathcal{L}_{I_{1}}\left(s,r_{1k}\right)m!}{\prod_{a=1}^{m}\left(p_a!(a!)^{m_{a}}\right)}\prod_{a=1}^{m}\left(-2\pi\lambda_{1}\int_{r_{1k}}^{\infty}\left(-\frac{(-1)^{a}\Gamma\left(1+a\right)}{r_{1k}^{a\alpha_{1}}\left(1+s r_{1k}^{-\alpha_{1}}\right)^{a+1}}\right)r{\rm d}r\right)^{p_{a}}
\end{align}}where the integral can be solved using similar method as calculating (\ref{eq:LT}). Similarly, we can calculate $\mathcal{L}_{I_{2}}\left(s,r_{2k}\right)$ and its higher order derivative $\mathcal{L}_{I_{2}}^{(m)}\left(s,r_{2k}\right)$. Finally, after some algebraic manipulations, we can obtain $\mathcal{S}_{{\rm IN},k,R_{1k},R_{2k}}(r_{1k},r_{2k},\beta)$ where $k\in\{1,2\bar{O},2OC\}$.

\subsubsection{$k=2O\bar{C}$}
When $k=2O\bar{C}$, based on (\ref{eq:SIR_BbarC_v2}), using multinomial theorem, and following similar procedures in calculating (\ref{eq:CPcondi_twoIterms}), 
we can obtain $\mathcal{S}_{{\rm IN}, k,R_{1k},R_{2k}}(r_{1k},r_{2k},\beta)$.

\subsection{Proof of Lemma \ref{lem:DeltaS_sign}}\label{proof:DeltaS_sign}
\subsubsection{Proof of $\Delta\mathcal{\bar{R}}_{{\rm IN},1}(U,\tau)<0$}
When the design parameter is $U$, we have 
{\small\begin{align}\label{eq:CP_macro_U}
\mathcal{\bar{R}}_{{\rm IN},1}(U,\tau)=\int_{0}^{\infty}\left(\sum_{u=0}^{U}\left(\sum_{n=0}^{N_1-u-1}\mathcal{T}_{1,Y_{1}}(n,y,\hat{\beta})\right){\rm Pr}\left(u_{2OC,0}\left(U\right)=u\right)\right)f_{Y_{1}}(y){\rm d}y
\end{align}}where {\small${\rm Pr}\left(u_{2OC,0}\left(U\right)=u\right)=
\begin{cases}
&{\rm Pr}\left(U_{2O_{a},0}=u\right),\hspace{1.05cm} {\rm for}\; 0\le u<U\\
&\sum_{u=U}^{\infty} {\rm Pr}\left(U_{2O_{a},0}=u\right), \quad {\rm for}\; u=U
\end{cases}$}, and {\small$\hat{\beta}=2^{\frac{{\rm E}\left[L_{0,1}\right]\tau}{W}}-1$}.

Similarly, when the design parameter is $U-1$, we have 
{\small\begin{align}\label{eq:CP_macro_U-1}
\mathcal{\bar{R}}_{{\rm IN},1}(U-1,\tau)=\int_{0}^{\infty}\left(\sum_{u=0}^{U-1}\left(\sum_{n=0}^{N_1-u-1}\mathcal{T}_{1,Y_{1}}(n,y,\hat{\beta})\right){\rm Pr}\left(u_{2OC,0}\left(U-1\right)=u\right)\right)f_{Y_{1}}(y){\rm d}y
\end{align}}where {\small${\rm Pr}\left(u_{2OC,0}\left(U-1\right)=u\right)=
\begin{cases}
&{\rm Pr}\left(U_{2O_{a},0}=u\right),\hspace{1.05cm} {\rm for}\; 0\le u<U-1\\
&\sum_{u=U-1}^{\infty} {\rm Pr}\left(U_{2O_{a},0}=u\right), \quad {\rm for}\; u=U-1
\end{cases}$}. 

Based on (\ref{eq:CP_macro_U}) and (\ref{eq:CP_macro_U-1}), and after some algebraic manipulations, we have
{\small\begin{align}\label{eq:deltaS1_general}
\Delta\mathcal{\bar{R}}_{{\rm IN},1}(U,\tau)
=&-\left(1-\sum_{u=0}^{U-1}{\rm Pr}\left(U_{2O_{a},0}=u\right)\right)\int_{0}^{\infty}\mathcal{T}_{1,y}(N_{1}-U,y,\hat{\beta})f_{Y_{1}}(y){\rm d}y<0\;.
\end{align}}\subsubsection{Proof of $\Delta\mathcal{\bar{R}}_{{\rm IN},2\bar{O}}(\tau)=0$} Follows by noting that $\mathcal{\bar{R}}_{{\rm IN},2\bar{O}}(\tau)$ is independent of $U$.
\subsubsection{Proof of $\Delta\mathcal{\bar{R}}_{{\rm IN},2O}(U,\tau)>0$}
We first show that 
${\rm Pr}\left(\mathcal{E}_{2OC,0}(U-1)\right)<{\rm Pr}\left(\mathcal{E}_{2OC,0}(U)\right)$, which is as follows:
{\small\begin{align}
{\rm Pr}\left(\mathcal{E}_{2OC,0}(U-1)\right)&=\sum_{n=1}^{U-1}{\rm Pr}\left(\hat{U}_{2O_{a},0}=n\right)+\sum_{n=U}^{\infty}\frac{U-1}{n}{\rm Pr}\left(\hat{U}_{2O_{a},0}=n\right)\notag\\
&<\sum_{n=1}^{U-1}{\rm Pr}\left(\hat{U}_{2O_{a},0}=n\right)+{\rm Pr}\left(\hat{U}_{2O_{a},0}=U\right)+\sum_{n=U+1}^{\infty}\frac{U}{n}{\rm Pr}\left(\hat{U}_{2O_{a},0}=n\right)\notag\\
&={\rm Pr}\left(\mathcal{E}_{2OC,0}(U)\right)
\end{align}}where the inequality is obtained by noting that $\frac{U-1}{n}<\frac{U}{n}$ ($n\in \mathbb{N}$). Next, we show that $\mathcal{\bar{R}}_{{\rm IN},2OC}(U,\tau)>\mathcal{\bar{R}}_{{\rm IN},2O\bar{C}}(U,\tau)$. From (\ref{eq:SINR_BC}) and (\ref{eq:SIR_BbarC_v2}), we note that for any network and channel realizations, 
since $\frac{\left|\mathbf{h}_{1,10}^{\dagger}\mathbf{f}_{1,1}\right|^{2}}{Y_{1}^{\alpha_{1}}}>0$, we always have ${\rm SIR}_{{\rm IN},2OC,0}>{\rm SIR}_{{\rm IN},2O\bar{C},0}$. Hence, we have ${\rm Pr}\left({\rm SIR}_{{\rm IN},2OC,0}>\beta\right)>{\rm Pr}\left({\rm SIR}_{{\rm IN},2O\bar{C},0}>\beta\right)$, i.e., $\mathcal{\bar{R}}_{{\rm IN},2OC}(U,\tau)>\mathcal{\bar{R}}_{{\rm IN},2O\bar{C}}(U,\tau)$. Finally, since {\small$\Delta\mathcal{\bar{R}}_{{\rm IN},2O}(U,\tau)=\left({\rm Pr}\left(\mathcal{E}_{2OC,0}(U)\right)-{\rm Pr}\left(\mathcal{E}_{2OC,0}(U-1)\right)\right)\left(\mathcal{\bar{R}}_{{\rm IN},2OC}(\tau)-\mathcal{\bar{R}}_{{\rm IN},2O\bar{C}}(\tau)\right)$}, we obtain $\Delta\mathcal{\bar{R}}_{{\rm IN},2O}(U,\tau)>0$.


\subsection{Proof of Lemma \ref{lem:condiCP_lowbeta}}\label{proof:Tn_lowbeta}
Firstly, let {\small$\hat{\beta}=2^{\frac{{\rm E}\left[L_{0,j_{k}}\right]\tau}{W}}-1$}. 
It can be easily seen that $\hat{\beta}\to0$ when $\tau\to0$. Then, we investigate the asymptotic behavior of {\small$\mathcal{T}_{k,R_{1k},R_{2k}}\left(n,r_{1k},r_{2k},\hat{\beta}\right)$} when $\hat{\beta}\to0$. We note that 
\begin{align}
B^{'}(a,b,z)=\frac{(1-z)^{b}}{b}+o\left((1-z)^{b}\right)\;,\quad{\rm as}\; z\to1\;.
\end{align}
Then, we have 
{\small\begin{align}
B^{'}\left(\frac{2}{\alpha},1-\frac{2}{\alpha},\frac{1}{1+c\hat{\beta}}\right)&= \frac{\left(c\hat{\beta}\right)^{1-\frac{2}{\alpha}}}{1-\frac{2}{\alpha}}+o\left(\hat{\beta}^{1-\frac{2}{\alpha}}\right)\;,\\
B^{'}\left(1+\frac{2}{\alpha},a-\frac{2}{\alpha},\frac{1}{1+c\hat{\beta}}\right)&= \frac{(c\hat{\beta})^{a-\frac{2}{\alpha}}}{a-\frac{2}{\alpha}}+o\left(\hat{\beta}^{a-\frac{2}{\alpha}}\right)\;,
\end{align}}where $c\in\mathbb{R}^{+}$. Based on these two asymptotic expressions, and let $s=\tilde{c}\hat{\beta}$ ($\tilde{c}\in\mathbb{R}^{+}$) in $\mathcal{L}_{I_{j}}\left(s,r_{jk}\right)$ and $\mathcal{L}^{(m)}_{I_{j}}\left(s,r_{jk}\right)$, we can obtain 
{\small\begin{align}
\mathcal{L}_{I_{j}}\left(s,r_{jk}\right)&=1-\frac{2\pi\lambda_{j}\tilde{c}r_{jk}^{2-\alpha_{j}}}{\alpha_{j}\left(1-\frac{2}{\alpha_{j}}\right)}\hat{\beta}+o\left(\hat{\beta}\right)\;,\\ 
\mathcal{L}_{I_{j}}^{(m)}\left(s,r_{jk}\right)&=\hat{\beta}^{m}\left(\frac{\tilde{c}}{r_{jk}^{\alpha_{j}}}\right)^{m}\sum_{(p_{a})_{a=1}^{m}\in\mathcal{M}_{m}}\frac{m!}{\prod_{a}^{m}(p_{a}!)}\prod_{a=1}^{m}\left(\frac{2\pi\lambda_{j}r_{jk}^{2}}{\alpha_{j}\left(a-\frac{2}{\alpha_{j}}\right)}\right)^{p_{a}}+o\left(\hat{\beta}^{m}\right)\;. 
\end{align}}Moreover, when $u_{0}\in\mathcal{U}_{2O\bar{C}}$, we have {\small$\left(1+\hat{\beta}\frac{P_{1}Y_{2}^{\alpha_{2}}}{P_{2}Y_{1}^{\alpha_{1}}}\right)^{-\left(q_{3}+1\right)}=1-(q_{3}+1)\frac{P_{1}Y_{2}^{\alpha_{2}}}{P_{2}Y_{1}^{\alpha_{1}}}\hat{\beta}+o(\hat{\beta})$}. Substituting the series expansions of {\small$\mathcal{L}_{I_{j}}\left(s,r_{jk}\right)$}, {\small$\mathcal{L}^{(m)}_{I_{j}}\left(s,r_{jk}\right)$}, and {\small$\left(1+\hat{\beta}\frac{P_{1}Y_{2}^{\alpha_{2}}}{P_{2}Y_{1}^{\alpha_{1}}}\right)^{-\left(q_{3}+1\right)}$} into $\mathcal{T}_{k,R_{1k},R_{2k}}\left(n,r_{1k},r_{2k},\hat{\beta}\right)$, and after some algebraic manipulations, we have the final result.

\subsection{Proof of Proposition \ref{prop:CPgain_offloadIN}}\label{proof:gain_offload_lowbeta}
When $u_{0}\in\mathcal{U}_{2OC}$, since {\small$\left|\mathbf{h}_{2,00}^{\dagger}\mathbf{f}_{2,0}\right|^{2}\dis{\rm Gamma}(N_{2},1)$}, we have
{\small\begin{align}
1-\mathcal{\bar{R}}_{{\rm IN},2OC}(\tau)
&=\tau^{N_{2}}\frac{\left({\rm E}\left[L_{0,2}\right]\ln(2)\right)^{N_{2}}}{W^{N_{2}}N_{2}!}{\rm E}\left[Y_{2}^{\alpha_{2}}\left(\frac{P_{1}}{P_{2}}I_{1}+I_{2}\right)^{N_{2}}\right]+o\left(\tau^{N_{2}}\right)\notag\\
&=\tau^{N_{2}}\frac{\left({\rm E}\left[L_{0,2}\right]\ln(2)\right)^{N_{2}}}{W^{N_{2}}N_{2}!}\sum_{n=0}^{N_{2}}\binom{N_{2}}{n}\left(\frac{P_{1}}{P_{2}}\right)^{n}{\rm E}\left[\left(Y_{2}^{\alpha_{2}}I_{1}\right)^{n}\right]{\rm E}\left[\left(Y_{2}^{\alpha_{2}}I_{2}\right)^{N_{2}-n}\right]+o\left(\tau^{N_{2}}\right)\;.
\end{align}}In order to show that $1-\mathcal{\bar{R}}_{{\rm IN},2OC}(\tau)=\Theta\left(\tau^{N_{2}}\right)$, we need to show that {\small${\rm E}\left[\left(Y_{2}^{\alpha_{2}}I_{1}\right)^{n}\right]<\infty$} and {\small${\rm E}\left[\left(Y_{2}^{\alpha_{2}}I_{2}\right)^{N_{2}-n}\right]<\infty$}. This can be proved by noting that {\small${\rm E}\left[\left(Y_{2}^{\alpha_{2}}I_{2}\right)^{N_{2}-n}\right]<\infty$} \cite{haenggi14}, and {\small${\rm E}\left[\left(Y_{2}^{\alpha_{2}}I_{1}\right)^{n}\right]\stackrel{(a)}{<}\left(\frac{BP_{2}}{P_{1}}\right)^{n}{\rm E}\left[\left(Y_{1}^{\alpha_{1}}I_{1}\right)^{n}\right]<\infty$} \cite{haenggi14} where (a) is obtained by following {\small$Y_{2}^{\alpha_{2}}<\frac{BP_{2}}{P_{1}}Y_{1}^{\alpha_{1}}$} when $u_{0}\in\mathcal{U}_{2OC}$. Similarly, when $u_{0}\in\mathcal{U}_{2O\bar{C}}$, we have
{\small\begin{align}\label{eq:OPBbarC_lowbeta}
1-\mathcal{\bar{R}}_{{\rm IN},2O\bar{C}}(\tau)&=\tau^{N_{2}}\frac{\left({\rm E}\left[L_{0,2}\right]\ln(2)\right)^{N_{2}}}{W^{N_{2}}N_{2}!}{\rm E}\left[Y_{2}^{\alpha_{2}}\left(\frac{P_{1}}{P_{2}}I_{1}+I_{2}+\frac{1}{Y_{1}^{\alpha_{1}}}\frac{P_{1}}{P_{2}}g_{1,1}\right)^{N_{2}}\right]+o\left(\tau^{N_{2}}\right)=\Theta\left(\tau^{N_{2}}\right)\;.
\end{align}}Finally, by noting that {\small$\frac{1}{N_{2}!}{\rm E}\left[Y_{2}^{\alpha_{2}}\left(\frac{P_{1}}{P_{2}}I_{1}+I_{2}+\frac{1}{Y_{1}^{\alpha_{1}}}\frac{P_{1}}{P_{2}}g_{1,1}\right)^{N_{2}}\right]>\frac{1}{N_{2}!}{\rm E}\left[Y_{2}^{\alpha_{2}}\left(\frac{P_{1}}{P_{2}}I_{1}+I_{2}\right)^{N_{2}}\right]$}, we have $\mathcal{\bar{R}}_{{\rm IN},2OC}(\tau)-\mathcal{\bar{R}}_{{\rm IN},2O\bar{C}}(\tau)=\Theta\left(\tau^{N_{2}}\right)$. Moreover, since {\small${\rm Pr}\left(\mathcal{E}_{2OC,0}(U)\right)-{\rm Pr}\left(\mathcal{E}_{2OC,0}(U-1)\right)$} is independent of $\tau$, we obtain the final result.  


%
%

\subsection{Proof of Lemma \ref{lem:comp_INABS_eachuser}}\label{proof:compABS_CPrate}
\subsubsection{Proof of i)}
First, assuming that $N_{1}-U$ DoF are used for IN, we obtain a lower bound of $\mathcal{\bar{R}}_{{\rm IN},1}(U,\tau)$, denoted as $\mathcal{\bar{R}}^{\rm lb}_{{\rm IN},1}(U,\tau)$. Following similar procedures in \cite[Appendix B]{wu12}, we have the following result for $\mathcal{\bar{R}}^{\rm lb}_{{\rm IN},1}(U,\tau)$ when $N_{1},U\to\infty$ with $\frac{U}{N_{1}}\to\kappa\in(0,1)$ and $\tau\to0$:  
{\small\begin{align}\label{eq:CPrate1_IN}
\mathcal{\bar{R}}_{{\rm IN},1}^{\rm lb}(U,\tau)&\approx{\rm Pr}\left(\frac{\frac{P_{1}}{Y_{1}^{\alpha_{1}}}\frac{N_{1}\left(1-\kappa\right)}{2^{\frac{\tau}{W} {\rm E}\left[L_{0,1}\right]}-1}}{P_{1}I_{1}+P_{2}I_{2}}>1\right)\appa {\rm Pr}\left(\frac{P_{1}}{Y_{1}^{\alpha_{1}}\left(P_{1}I_{1}+P_{2}I_{2}\right)}>\frac{{\rm ln}(2)\tau{\rm E}\left[L_{0,1}\right]}{WN_{1}\left(1-\kappa\right)}\right)
\end{align}}where $(a)$ is obtained by noting that $2^{\frac{\tau}{W}{\rm E}\left[L_{0,1}\right]}\approx1+{\rm ln}(2){\rm E}\left[L_{0,1}\right]\frac{\tau}{W}$ as $\tau\to0$.


Similarly, for ABS, when $\tau\to0$, we have the following:
{\small\begin{align}\label{eq:CPrate1_ABS}
\mathcal{\bar{R}}_{{\rm ABS},1}(\tau)\approx {\rm Pr}\left(\frac{P_{1}}{Y_{1}^{\alpha_{1}}\left(P_{1}I_{1}+P_{2}I_{2}\right)}>\frac{{\rm ln}(2)\tau{\rm E}\left[L_{0,1}\right]}{WN_{1}\left(1-\eta\right)}\right)\;.
\end{align}}From (\ref{eq:CPrate1_IN}) and (\ref{eq:CPrate1_ABS}), we see that $\mathcal{\bar{R}}_{{\rm IN},1}^{\rm lb}(U,\tau)>\mathcal{\bar{R}}_{{\rm ABS},1}(\tau)$, which is a sufficient condition of $\mathcal{\bar{R}}_{{\rm IN},1}(U,\tau)>\mathcal{\bar{R}}_{{\rm ABS},1}(\eta,\tau)$, i.f.f. $\frac{1}{1-\kappa}<\frac{1}{1-\eta}$. After some manipulations, we have the final result. 

\subsubsection{Proof of ii)}
The proof is similar to that of i). 

\subsubsection{Proof of iii)}
We note that the SIR of $u_{0}\in\mathcal{U}_{2O}$ in the IN scheme is worse than that in ABS. Hence, in order to achieve $\mathcal{\bar{R}}_{{\rm IN},2O}(U,\tau)>\mathcal{\bar{R}}_{{\rm ABS},2O}(\eta,\tau)$, it is necessary that the average resource used to serve $u_{0}\in\mathcal{U}_{2O}$ in the IN scheme (i.e., $\frac{1}{{\rm E}\left[L_{0,2}\right]}$ under MLA) is lager than that in ABS (i.e., $\frac{\eta}{{\rm E}\left[L_{0,2O}\right]}$ under MLA). 

\end{document}